\documentclass[11pt]{article}

\usepackage{amssymb,amsmath,amsfonts,eurosym,geometry,ulem,graphicx,color,setspace,sectsty,comment,footmisc,caption,pdflscape,array,placeins,adjustbox,booktabs,longtable,dirtytalk,tabularx,multirow,soul,threeparttable,float,rotating, authblk, bbm, endnotes,tikz,listings,xcolor}
\usepackage{csquotes}
\usepackage{enumitem}
\usetikzlibrary{arrows.meta,positioning,calc}

\usepackage{tcolorbox}
\usepackage[hyperfootnotes=false, colorlinks=true, linkcolor=black, urlcolor=black, citecolor={blue}]{hyperref}

\newtcolorbox{llmthought}{
  colback=black!5,
  colframe=black!20,
  boxrule=0.5pt,
  arc=2pt,
  left=10pt,
  right=10pt,
  top=8pt,
  bottom=8pt,
  fontupper=\ttfamily\footnotesize,
  title={Quotes from the LLM chain-of-thought},
  fonttitle=\sffamily\small,
  coltitle=black
}

\lstdefinestyle{promptstyle}{
  basicstyle   = \ttfamily\scriptsize,
  mathescape   = true,
  frame        = single,
  breaklines   = true,
  breakatwhitespace = true,
  postbreak    = \mbox{$\hookrightarrow$\space},
  escapeinside = {(*@}{@*)}          
}

\usepackage{natbib}
 \bibpunct[, ]{(}{)}{,}{a}{}{,}%
 \def\bibsep{\smallskipamount}%

\newtheorem{proposition}{Proposition}
\newenvironment{proof}[1][Proof]{\noindent\textbf{#1.} }{\ \rule{0.5em}{0.5em}}

\newcolumntype{L}[1]{>{\raggedright\let\newline\\arraybackslash\hspace{0pt}}m{#1}}
\newcolumntype{C}[1]{>{\centering\let\newline\\arraybackslash\hspace{0pt}}m{#1}}
\newcolumntype{R}[1]{>{\raggedleft\let\newline\\arraybackslash\hspace{0pt}}m{#1}}

\begin{document}

\title{On the Fragility of AI Agent Collusion}



\author{Jussi Keppo, Yuze Li, Gerry Tsoukalas, Nuo Yuan\thanks{Keppo (\href{mailto:keppo@nus.edu.sg}{keppo@nus.edu.sg}): National University of Singapore; Li (\href{mailto:yuzeli@bu.edu}{yuzeli@bu.edu}), Tsoukalas (\href{mailto:gerryt@bu.edu}{gerryt@bu.edu}), Yuan (\href{mailto:nyuan01@bu.edu}{nyuan01@bu.edu}): Boston University.}}


\date{January 30, 2026\footnote{First version: September 25, 2025}}

\maketitle

\begin{abstract}

\singlespacing



Recent work shows that pricing with symmetric LLM agents leads to algorithmic collusion. We show that  collusion is fragile under the heterogeneity typical of real deployments. In a stylized repeated-pricing model, heterogeneity in patience or data access reduces the set of collusive equilibria. Experiments with open-source LLM agents (totaling over 2,000 compute hours) align with these predictions: patience heterogeneity reduces price lift from 22\% to 10\% above competitive levels; asymmetric data access, to 7\%. Increasing the number of competing LLMs breaks up collusion; so does cross-algorithm heterogeneity, that is, setting LLMs against Q-learning agents. But model-size differences (e.g., 32B vs. 14B weights) do not; they generate leader-follower dynamics that stabilize collusion. We discuss antitrust implications, such as enforcement actions restricting data-sharing and policies promoting algorithmic diversity.


\bigskip



\end{abstract}

\setstretch{1.45}

\section{Introduction}\label{sec: introduction}

Firms are increasingly delegating pricing decisions to algorithms, making algorithmic collusion a central concern for antitrust authorities \citep{gans2023pricecollusion}. Large language models (LLMs) push this trend further by enabling agents that interpret context and revise strategy autonomously. For example, Anthropic's \textit{Project Vend} gave Claude price-setting authority in a small but real-world retail setting \citep{anthropic2025projectvend}. On a larger scale, Delta Airlines announced that it is using generative AI to determine a portion of its domestic ticket prices \citep{mccarthy2025delta}. Against this backdrop, a senior U.S. Department of Justice official recently warned this trend could lead to  ``fully automated cartels operating without any human involvement'' \citep{steptoe2025speech}.

A growing literature shows that autonomous pricing algorithms can sustain tacit collusion without explicit communication \citep{calvano2020, klein2021, banchio2022artificial}. More recently, \citet{fish2024} show that LLM-based agents also collude, may converge faster than reinforcement learning (RL) approaches, and do so robustly across prompt variations. Yet existing studies typically assume identical or near-identical agents competing in highly stylized markets. Real deployments are messier:  firms adopt models that have different capabilities, different algorithms, proprietary prompts, varying update frequencies, and asymmetric information sets \citep{Huang2023ChatGPTGig}.  This raises a central question for competition policy: \emph{does LLM-mediated algorithmic collusion survive the heterogeneity typical of real-world deployments, and if variation disrupts collusion, which dimensions matter most?} This paper provides systematic evidence.

\textbf{Research Methodology:} Answering this question with field data is currently challenging; LLM pricing deployments are too recent to generate public observations at scale. We therefore adopt a two-pronged approach that leverages economic theory on one hand and laboratory experiments in the spirit of \citet{fish2024} on the other. The two play complementary roles. Experiments allow us to manipulate dimensions of heterogeneity independently---something observational data, even if available, would rarely permit. But the space of possible variation is large (models differ, prompts differ, information sets differ, etc.). Experimentation without theoretical grounding risks being ad hoc. Theory helps guide our empirical design by identifying which dimensions have equilibrium consequences.

\textbf{Model:} The folk theorem establishes that tacit collusion is sustainable if and only if players are sufficiently patient, making the discount factor a central primitive \citep{fudenberg1986folk}. \citet{APS1986, APS1990} extend this logic to environments with imperfect monitoring, showing that the ability to detect and punish deviations depends critically on the information available to players. These two primitives---patience and information---anchor our initial analysis. We develop a simple infinite-horizon pricing model that allows both to vary across agents. The model yields two comparative statics predictions. First, heterogeneity in patience undermines collusion: when discount factors differ, the impatient agent faces stronger deviation incentives, making coordination on supra-competitive prices harder to sustain. Second, asymmetric access to market data raises the patience threshold required for collusion, because the informationally disadvantaged agent cannot reliably detect deviations and therefore cannot credibly threaten punishment. The model is intentionally stylized; its contribution is not standalone theory but a principled framework linking heterogeneity dimensions to testable predictions. Our approach follows the methodology of behavioral game theory \citep{camerer2003behavioral}---using equilibrium predictions as benchmarks for evaluating actual behavior without claiming agents compute equilibria explicitly.

\textbf{Experimental Design:} Our experimental design deliberately relaxes the model's stylized assumptions to stress test whether its theoretical predictions survive harsher, more realistic deployment frictions. Agents in real deployments lack direct knowledge of demand primitives, cost structures, and equilibrium concepts. On the other hand, LLMs are trained on corpora that include textbook treatments of competition and oligopoly, knowledge they may draw on without explicit instruction. We implement a repeated Bertrand game using agents based on DeepSeek-R1, an open-source reasoning model whose weights and architecture are publicly available. Open-source selection is not incidental: it ensures full replicability and addresses concerns that results might depend on proprietary, unverifiable model internals. Each period, agents observe their own price, realized demand, and profit, but not the demand function or structural primitives. Patience is implemented via the objective specified in each agent's prompt: agents maximize profits over varying time horizons. Information asymmetry is implemented by varying access to competitors' historical prices. The design tests whether prompt-specified configurations translate into behavior consistent with theory, and whether heterogeneity across these dimensions disrupts coordination.

A natural question is whether our experiments identify the cognitive mechanism by which LLM agents collude. They do not. By design. LLMs are black boxes: even frontier diagnostic methods (such as mechanistic interpretability \citealt{bricken2023monosemanticity, narad2025mechanistic}) cannot fully characterize how these models map inputs to outputs, and firms deploying such systems have no more insight into their workings than researchers do. We can (and do) examine agents' chain-of-thought outputs for suggestive patterns, but these are not mechanistic explanations. This limitation, however, does not erase the paper's contribution. From an antitrust and welfare perspective, what harms consumers is elevated prices, not the cognitive process producing them. The same epistemological constraint applies to human collusion: we cannot verify whether executives coordinate through game-theoretic reasoning or heuristics acquired through experience, yet tacit collusion enforcement proceeds by targeting conduct, not cognition. Outcome measurement---whether collusive pricing emerges and persists---is the policy-relevant quantity and, given LLM opacity, the only currently feasible diagnostic. Our experiments are designed accordingly.

We assess collusive outcomes using two criteria: price convergence (whether agents coordinate on stable pricing) and price elevation relative to the static Nash benchmark (the welfare-relevant measure of market power). This dual criterion distinguishes genuine collusion from competitive convergence.

\textbf{Results:} \textit{Sanity checks.} Before testing heterogeneity effects, we verify that our LLM agents respond to prompt-specified objectives in ways consistent with standard theory; absent this, observed price dynamics could reflect instruction-following failures rather than strategic responses to the environment. We confirm that agents recover the one-shot monopoly optimum and the static Bertrand--Nash equilibrium with high precision (within 5\% of theoretical benchmarks). In repeated settings, patient agents converge to supra-competitive pricing while myopic agents price competitively, with outcomes varying monotonically in the specified discount factor. These homogeneous extremes (two patient vs.\ two myopic) serve a dual role: they validate that the prompt-specified horizon shifts behavior in the expected direction, and they provide the baseline against which we measure how heterogeneity attenuates collusion.

Having established that prompt-specified objectives translate into theoretically coherent behavior, we test whether agent responses track the model's comparative-statics predictions, then extend the analysis to heterogeneity dimensions that lack clear theoretical analogues: cross-algorithm heterogeneity and model size. We also examine the implications of increasing the number of competing LLMs. A summary of these results can be found in Table~\ref{tab:results_at_a_glance}.

\begin{table}[ht]
\centering
\caption{Results at a Glance: Collusion Outcomes under Agent Heterogeneity.}
\label{tab:results_at_a_glance}
\small
\setlength{\tabcolsep}{6pt}
\begin{tabular}{lcc}
\hline
\textbf{Experiments} & \textbf{Collusion strength} & \textbf{Price elevation}\\
 & (Low/Med./High) & (\% vs.\ competitive) \\
\hline
\multicolumn{3}{l}{\textit{Homogeneous Benchmarks}} \\
Two myopic LLMs (no collusion emerges) & Low & $\approx 0\%$ \\
Two patient LLMs (tacit collusion emerges)  & High & +22\% \\
\hline
\multicolumn{3}{l}{\textit{Main Heterogeneity dimensions (predicted by model)}} \\
1 patient LLM agent vs.\ 1 myopic LLM agent & Medium & +10\% \\
1 high-data LLM agent vs.\ 1 low-data LLM agent & Medium & +7\% \\
\hline
\multicolumn{3}{l}{\textit{Additional heterogeneity dimensions}} \\
1 patient LLM agent vs.\ 1 Q-learning agent (frozen) & Low & +2\% \\
1 patient LLM agent vs.\ 1 Q-learning agent (adaptive) & Low & +8\%  \\
3 patient LLM agents  & High & +36\% \\
4 patient LLM agents  & Medium & +11\% \\
5 patient LLM agents  & Low & +3\% \\
1 patient 32B LLM agent vs.\ 1 patient 14B-LLM agent & High & +21\% \\
\hline
\end{tabular}
\end{table}

\textit{Testing model predictions.} Both theoretical predictions find support. Patience heterogeneity reduces collusive outcomes: when two patient agents interact, they converge to prices approximately 22\% above the static Nash equilibrium, with pricing rationales that explicitly emphasize avoiding destructive price wars. Setting a patient agent against a myopic one produces meaningfully weaker price elevation, with prices only about 10\% above competitive levels. Fully myopic pairs converge essentially to the competitive outcome. Chain-of-thought inspection indicates that the myopic agent's inability to credibly commit to future punishment weakens the patient agent's incentive to maintain cooperation. Information asymmetry also undermines cooperation: relative to the homogeneous-information benchmark, heterogeneous data access reduces the deviation from the competitive benchmark from 22\% down to 7\%, consistent with the model's prediction that information asymmetry makes tacit coordination more difficult.

\textit{Other Sources of Heterogeneity.} Cross-algorithm heterogeneity disrupts collusion entirely: setting Q-learning agents against LLM agents prevents sustained coordination. Across runs, they do not reach sustained collusion within 1,000 periods and instead remain near competitive pricing. The algorithms operate on incompatible signal spaces: Q-learning agents cannot interpret the cooperative pricing patterns that LLM agents attempt to establish, generating instability that drives both agents toward competitive pricing. 

Increasing the number of competing agents gradually, from two to five, delays collusion onset and weakens its robustness. With three agents competing, sustained collusion remains feasible and can lead to even higher prices, but coordination/convergence becomes  meaningfully slower. With four, prices drop, collusion is further slowed and becomes less stable. With five sellers, sustained collusion fails to emerge within the experimental horizon. These patterns largely align with classical oligopoly theory. 

Perhaps most surprisingly, model-scale heterogeneity does not impede collusion. When pairing larger (32B) and smaller (14B) LLM agents (both patient), final price elevation remains close to the homogeneous benchmark, though convergence is substantially slower. Novel dynamics emerge: the larger model assumes a leadership role and anchors the collusive price, while the smaller model adopts a follower position, suggesting that capability asymmetry can facilitate coordination by establishing clear roles.

Finally, additional robustness tests with a closed-source model (GPT-5-mini) serve two purposes: confirming that our main qualitative findings extend to a more capable model, and examining whether LLM agents exhibit sensitivity to explicit collusion-avoidance prompting---a dimension unique to language-model-based agents. Both hold.

\textbf{Policy Implications:} These findings speak to active policy debates that we discuss in detail in Section~\ref{sec:conclusion}. We provide a brief preview here. 

Our patience-heterogeneity results reveal a mechanism distinct from \citet{calvano2020}, who find collusion diminishes when all agents share a lower discount factor: our findings suggest diversity in strategic horizons itself disrupts coordination. Among other implications, this suggests merger reviews should consider whether transactions homogenize pricing objectives (e.g, capital discount rates), not just algorithms \citep{gal2024algorithms}. 

Our data-access results provide experimental support for the theory underlying the DOJ's RealPage case \citep{FedRegRealPage2026}, which targeted information-sharing as a mechanism enabling coordination. 

\citet{assad2024algorithmic} find that algorithmic pricing increases margins only when all local competitors adopt, suggesting homogeneity may be necessary for coordination. Our cross-algorithm heterogeneity findings offer a mechanism: pairing incompatible algorithmic architectures disrupts coordination entirely, as agents cannot parse each other's signaling strategies. Maintaining ecosystem diversity may therefore complement the platform-level interventions studied by \citet{johnson2023}.

Our market-structure results confirm that classical intuitions about competitor count \citep{stigler1964theory} extend to LLM agents. 

Finally, our capability-asymmetry result offers a cautionary note consistent with \citet{athey2001optimal}: policymakers should not assume that heterogeneous AI capabilities always disrupt coordination; these differences can instead stabilize collusion through leader-follower dynamics. 


\subsection*{Related Literature}

The two closest papers are \citet{calvano2020} and \citet{fish2024}. \citet{calvano2020} show that homogeneous Q-learning agents can learn tacitly collusive pricing, and they document robustness to standard changes in the environment, including cost and demand asymmetries, stochastic demand, and changes in the number of firms. \citet{fish2024} show the corresponding result for LLM pricing agents. They compare prompt variants and show that while wording can affect the level of price elevation, supracompetitive outcomes persist; they also report robustness to demand-side asymmetry (unequal product qualities) and stochastic demand.

Our contribution is different. Even if collusion is relatively robust to prompt wording, real deployments differ along other dimensions, and these differences need not be benign. We therefore purposely hold fixed the baseline prompt structure and ask whether tacit collusion survives cross-firm and cross-algorithm differences. We find that several of these heterogeneities suffice to materially weaken or break sustained price elevation even without introducing additional prompt variability.

More broadly, our work connects three literatures: algorithmic collusion, firm heterogeneity, and platform governance.

On algorithmic collusion, a broader literature confirms that reinforcement learning agents generate supra-competitive outcomes \citep{klein2021, hansen2021, brown2025algorithmic, musolff2025algorithmicpricing}, with mechanisms including spontaneous coupling \citep{banchio2022artificial}.\footnote{Consumer-price implications can differ in richer environments; see \citet{zhao2025algorithmic} for ad auctions with search frictions and \citet{dou2025ai} for securities trading.} However, this literature also reveals sensitivity to implementation details: even minor differences such as exploration rules can shift outcomes \citep{asker2022, abada2023, wang2023}. This sensitivity motivates our focus on heterogeneity. Meanwhile, classical reinforcement learning faces scalability constraints \citep{deng2024, denboer2024}, driving practical interest toward LLMs, which adapt without retraining \citep{miller2024harnessing, tambe2025reskilling, shetty2025advanced}. The real-world deployments discussed above (Project Vend, Delta Airlines) exemplify this shift.

On firm heterogeneity, the classical literature offers competing predictions. More competitors increase deviation incentives \citep{stigler1964theory, green1984noncooperative, miller2017understanding, allon2019information}. Cost asymmetries create divergent incentives \citep{brock1985price, athey2008collusion}, size disparities weaken punishment credibility \citep{genesove2001rules}, and information asymmetries undermine focal equilibria \citep{laffont1997collusion, harrington2011private}. Yet structured heterogeneity can facilitate collusion through focal roles \citep{athey2001optimal, compte2002capacity, brown2023competition}. These insights derive from human decision-makers or simple algorithms. Our theoretical framework builds on the folk theorem logic of \citet{fudenberg1986folk} and the imperfect monitoring results of \citet{APS1986, APS1990}, as described above, to generate predictions for LLM agents.

On governance, recent work proposes platform-level interventions including price-directed prominence \citep{johnson2023}, anti-collusive mechanism design \citep{brero2022}, and unpersonalized rankings \citep{qiu2023}. These target Q-learning dynamics. As discussed, LLM agents pose distinct regulatory challenges because the relevant heterogeneity dimensions (model architecture, prompt design, data access) differ from those studied previously.

The remainder of the paper proceeds as follows. Section~\ref{sec:model} presents the theoretical framework. Section~\ref{sec:experimental_design} describes the experimental design. Sections~\ref{sec:intrinsic_characteristics} and \ref{sec: external_market_conditions} report results. Section~\ref{sec:robustness} conducts additional robustness tests. Section~\ref{sec:conclusion} concludes.

\section{Stylized Model}\label{sec:model}

This section presents a stylized infinite-horizon pricing model with two sellers. As discussed, the model serves as a qualitative benchmark to generate directional hypotheses, not as a structural description of LLM cognition. The model identifies which dimensions of heterogeneity should matter for collusion under canonical game-theoretic reasoning; our experiments then test whether LLM agents---deployed under realistic conditions without access to the model---produce outcomes directionally consistent with these predictions. Section \ref{subsec:model_symmetric} analyzes a baseline case with heterogeneous patience levels, and Section \ref{subsec: model_asymmetric} extends the model to a heterogeneous data access case.

\begin{table}[H]
\centering
\small
\renewcommand{\arraystretch}{0.85}
\caption{Summary of Key Notation}
\begin{tabular}{@{}cl@{}}
\toprule
\textbf{Variable} & \textbf{Definition}\\
\midrule
$N$ & Number of sellers \\
$a$ & Vertical differentiation parameter \\
$b$ & Own-price sensitivity parameter \\
$d$ & Cross-price sensitivity parameter \\
$\mu$ & Price elasticity parameter \\
$a_{0}$ & Outside option parameter \\
$p_i^t$ & Price chosen by seller $i$ in period $t$ \\
$\mathbf{p}^{t}$ & Vector of prices chosen by all $N$ sellers in period $t$ \\
$Q_i(\mathbf{p}^{t})$ & Demand of seller $i$ in period $t$ \\
$\delta_i \in [0,1)$ & Discount factor / patience level of seller $i$ \\
$\rho_i \in [0,1]$& Monitoring precision / data access of seller $i$ \\
$s^t \in \{G, B\}$ & Monitoring signal in period $t$\\
$c$ & Constant marginal cost \\
$\pi_i^t$ & Profit of seller $i$ in period $t$ \\
$\Pi_i$ & Discounted infinite-horizon profit of seller $i$ \\
\bottomrule
\end{tabular}
\label{table: key_notations}
\begin{minipage}{0.95\textwidth}
\vspace{0.5em}
\centering
\footnotesize
\textit{Note:} Non-bold symbols are scalars; bold symbols are vectors.
\end{minipage}
\end{table}

\subsection{Heterogeneous Patience Level}\label{subsec:model_symmetric}

Following the repeated-game framework in \cite{APS1986, APS1990}, we consider a market with two sellers each indexed by $i \in \{1, 2\}$, competing by setting prices in infinitely many discrete periods, indexed by $t = 0, 1, 2, \dots$. In each period, the sequence of events is as follows: (1) sellers simultaneously choose current period prices $p_i^t$; (2) prices are publicly observed and the corresponding market demands are realized; and (3) sellers earn profits based on the realized demands. Both sellers face identical constant marginal costs $c > 0$, and have symmetric linear demand. Specifically, seller $i$'s demand in period $t$ is given by
\begin{equation}\label{eq:linear_demand}
    Q_i(\mathbf{p}^{t}) = a - bp_i^t + dp_{-i}^t \, , 
\end{equation}
where $a > 0$ is the vertical differentiation parameter, $b > 0$ is the own-price sensitivity, and $d>0$ is the cross-price sensitivity. To ensure the existence of equilibrium, we posit the following assumptions: (1) $a> c(b-d)$, (2) $c<\frac{a}{b}$, and (3) $b>d > 0$. Seller $i$'s profit in period $t$ is then defined as
\begin{equation}\label{eq:profit_function}
    \pi_i^t = Q_i^t(p_i^t-c) \, .
\end{equation}

Each seller aims to maximize its discounted infinite-horizon profit given by
\begin{equation}\label{eq:discounted_profit_function}
    \Pi_i = \sum_{t=0}^{\infty} \delta_i^t \pi_i^t \, ,
\end{equation}
where $\delta_i \in [0,1)$ denotes the discount factor (patience level) of seller $i$. A higher $\delta_i$ means that seller $i$ is more patient, placing more weight on future profits, whereas a lower $\delta_i$ corresponds to a more myopic seller who cares primarily about immediate gains. We allow for heterogeneity in patience levels across sellers.

For simplicity, and in line with the simple bang-bang strategies in \cite{APS1986, APS1990}, we assume sellers employ the standard \textit{grim trigger strategy} \citep{friedman1971, dutta2007}. Specifically, seller $i$'s pricing strategy in period $t$ is defined as
\[
p_i^t = 
\begin{cases}
p^c, & \text{if $p_i^{t-1} = p_{-i}^{t-1} = p^c$, or if } t=0,\\[6pt]
p^*, & \text{otherwise,}
\end{cases}
\]
where $p^c \in (p^*, p^M]$ denotes the target price the sellers aim to collude on, $p^*$ denotes the one-shot Nash competitive equilibrium price, and $p^M > p^c$ denotes the one-shot joint-profit-maximizing monopoly price. It is important to note that $p^* = \frac{a+bc}{2b-d} > c$, so the one-shot Nash competitive equilibrium price is strictly higher than the marginal cost. Under the bang-bang pricing restriction, in each period, each seller's feasible price choice is limited to the two focal prices $\{ p^c, p^*\}$, and the one-shot deviation from the collusive price corresponds to switching from $p^c$ to $p^*$.

One consequence of our model assumptions is the following ordering of key payoff terms
\[
\pi(p^{*}_{i}, p^{c}_{-i}) > \pi^c > \pi^{*} > \pi(p^{c}_{i}, p^{*}_{-i}),
\]
where \(\pi(p^{*}_{i}, p^{c}_{-i})\) denotes the deviation payoff, representing the profit seller \( i \) obtains by deviating to the competitive price \(p^*\) while its competitor \(-i\) maintains the collusive price \(p^c\).
\(\pi^c\) denotes the collusive equilibrium payoff, realized by both sellers when they consistently maintain the collusive price 
\(p^c\). The term \(\pi^{*}\) denotes the competitive equilibrium payoff, achieved by both sellers when pricing at the competitive equilibrium price \(p^*\). 
The term \(\pi(p^{c}_{i}, p^{*}_{-i})\) denotes the payoff to seller \( i \) who maintains the higher collusive price \(p^c\) while being undercut by its competitor \(-i\), who chooses the competitive price \(p^*\).

Next, we identify conditions under which price collusion at a given $p^c$ can be sustained. In the following proposition, we characterize the threshold discount factor (i.e., the required level of patience among sellers) that enables stable and ongoing collusion, and how this threshold varies with the target collusive price.

\begin{proposition}\label{prop:symmetric_info_threshold}
    (Heterogeneous Patience Level Collusion Threshold) 
    For any given price $p^c \in (p^*, p^M]$:

    \begin{enumerate}
        \item [(i)] There exists a patience level threshold $\bar \delta(p^c) = \frac{\pi(p^{*}_{i}, p^{c}_{-i})- \pi^c}{\pi(p^{*}_{i}, p^{c}_{-i})- \pi^*}$ 
    such that two sellers
    can sustain collusion at price $p^c$ infinitely if and only if $\delta_i \geq \bar{\delta}(p^c)$ for $i = 1,2$.

    \item [(ii)] The threshold $\bar{\delta}(p^c)$ is strictly increasing in $p^c$.
    \end{enumerate} 
\end{proposition}

Part (i) establishes that sustaining collusion requires both sellers to be sufficiently patient ($\delta_i \geq \bar \delta$). When patient enough, the threat of future punishment outweighs the short-term gain from deviation; otherwise, collusion breaks down. Part (ii) shows that the patience threshold increases in the target collusive price $p^c$: as $p^c$ rises, deviation becomes more attractive because the rival's higher price shifts demand toward the deviating seller.

\subsection{Heterogeneous Data Access}\label{subsec: model_asymmetric}


We extend the model to an imperfect monitoring environment that captures differences in sellers' ability to observe rivals' pricing histories, which provide a stylized benchmark for the heterogeneity in data access experiment we run later in Section \ref{subsec: information_env_heterogeneity}.

The economic environment is as in Section \ref{subsec:model_symmetric}: two sellers with patience $\delta_i$ compete in an infinitely repeated Bertrand game with linear demands \eqref{eq:linear_demand}, identical marginal cost $c$, and objectives \eqref{eq:discounted_profit_function}. The difference is that sellers vary in the precision with which they detect deviations, which we interpret as differences in data access.

Formally, after prices are chosen in period $t$, sellers observe a public monitoring signal $s^t \in \{G, B\}$ about the pricing behaviors of sellers in that period. If seller $-i$ deviates from the collusive price $p^c$ in period $t$, then seller $i$ detects the deviation with probability $\rho_i \in [0,1]$. If the deviation is detected, a bad public signal $s^t = B$ is realized and observed by both sellers; otherwise, the public signal remains $s^t = G$.
If both sellers price at the collusive price $p^c$ in period $t$, then the public signal is always $s^t = G$. That is, we assume no false alarms for simplicity.  Therefore, the parameter $\rho_i$  represents seller $i$'s monitoring precision or data access: $\rho_i = 1$ corresponds to perfect observation / full access of the rival's price (as in the case of Section \ref{subsec:model_symmetric}), while $\rho_i <1$ corresponds to limited or noisy data access, with a lower $\rho_i$ corresponding to more limited or noisy data access.

We continue to assume that sellers employ the standard grim-trigger strategy, adapted to this imperfect monitoring environment. As in Section \ref{subsec:model_symmetric}, we maintain the bang-bang pricing restriction that each seller's feasible price choice is limited to $\{p^c, p^*\}$. Specifically, each seller $i$ starts by setting the target collusive price $p^c$ in period $t = 0$ and continues to set $p^c$ as long as a bad signal has never been observed. Once a bad signal $s^\tau = B$ is observed in some period $t = \tau$, both sellers permanently revert to the competitive equilibrium price $p^*$ in all subsequent periods. Formally, seller $i$'s pricing strategy in period $t$ can be written as 
\[
p_i^t = 
\begin{cases}
p^c, & \text{if $s^\tau = G$ for all $\tau < t$,}\\[6pt]
p^*, & \text{otherwise,}
\end{cases}
\]
where $p^c \in (p^*, p^M]$ denotes the target price the sellers aim to collude on, $p^* > c$ denotes the one-shot competitive equilibrium price, and $p^M > p^c$ denotes the one-shot joint-profit-maximizing monopoly price. 

As before, let \(\pi^c\) denote the collusive equilibrium payoff, \(\pi^{*}\) the competitive equilibrium payoff, \(\pi(p^{*}_{i}, p^{c}_{-i})\) the deviation payoff to seller \( i \) from undercutting its rival at \(p^c\) by moving to \(p^*\), and \(\pi(p^{c}_{i}, p^{*}_{-i})\) the payoff to seller \( i \) who maintains the higher collusive price \(p^c\) while being undercut by its competitor \(-i\) who deviates to the competitive price \(p^*\). We have the same payoff ordering as in Section \ref{subsec:model_symmetric}:
\[
\pi(p^{*}_{i}, p^{c}_{-i}) > \pi^c > \pi^{*} > \pi(p^{c}_{i}, p^{*}_{-i}).
\]
Next, we identify conditions under which price collusion at a given $p^c$ can be sustained under imperfect monitoring. In the following proposition, we characterize the threshold discount factor (i.e., the required level of patience among sellers) that enables stable and ongoing collusion, and how this threshold varies with the monitoring precision.

\begin{proposition}\label{prop:asymmetric_info_threshold}
    (Heterogeneous Data Access Collusion Threshold) For any given price $p^c \in (p^*, p^M]:$ 
    
    \begin{enumerate}
    \item[(i)] For each seller $i$, there exists a patience level threshold
    $\bar{\delta}_i(p^c,\rho_{-i})
    = \frac{\pi(p_i^*,p_{-i}^c)-\pi^c}
    {\pi(p_i^*,p_{-i}^c)-\pi^c + \rho_{-i}(\pi^c-\pi^*)}$
    such that two sellers can sustain collusion at price $p^c$ infinitely if and only if $\delta_i \;\ge\; \bar{\delta}_i(p^c, \rho_{-i})$ for $i = 1, 2$.
    
    \item[(ii)]
    The threshold $\bar{\delta}_i(p^c, \rho_{-i})$ is strictly decreasing in the rival's monitoring precision $\rho_{-i}$.
\end{enumerate}

\end{proposition}



Part (i) indicates that under imperfect monitoring, sustaining collusion still requires sufficient patience, but the threshold now depends on the rival's monitoring precision. With imperfect monitoring, a deviation by seller $i$ triggers punishment only if detected, which occurs with probability $\rho_{-i}$.

Part (ii) shows that the patience threshold $\bar{\delta}_i(p^c, \rho_{-i})$ is strictly decreasing in $\rho_{-i}$. When the rival can monitor deviations accurately (high $\rho_{-i}$), deviations are likely detected and punished, lowering the patience needed to sustain cooperation. When the rival has poor data access (low $\rho_{-i}$), deviating sellers can often escape detection, making deviation more attractive.

\section{Experimental Design}\label{sec:experimental_design}

This section describes our experimental design, which will allow us not only to test Propositions 
1 and 2, but also conduct heterogeneity experiments that go beyond the scope of the model. Before proceeding, we clarify the relationship between theory 
and experiment.

\subsection{Mapping Theoretical Parameters to Experimentation}\label{sec: map}

As mentioned, our experimental design deliberately withholds the model and its stylized assumptions from the agents. Agents in practice do not know precise demand primitives, cost structures, or equilibrium concepts; they learn from their pre- and post-training and information feedback loops with the user and the environment they are deployed in. Our experiments aim to preserve this realism: agents observe only their own prices, realized demand, and profits, plus (in some treatments) rivals' price histories. 

We emphasize (again) the model's purpose is to identify which dimensions of heterogeneity should matter under canonical game-theoretic reasoning; the experiments test whether LLM agents produce directionally consistent outcomes in more realistic deployment environments, despite lacking access to this structure.

Before detailing the experimental setup, we discuss how stylized modeling parameters can be interpreted in the context of more realistic deployments. Table~\ref{tab:mapping} summarizes the discussion.

\textbf{Patience.} The discount factor $\delta_i$ maps to the optimization horizon, which is directly specified in each agent's prompt. It is a governance decision firms make when configuring agents. How LLMs internally process this instruction is unknown, but for our purposes irrelevant: we test whether prompt-specified horizons produce outcomes directionally consistent with theoretical predictions, not whether agents implement the theory's logic.

\textbf{Monitoring precision.} The parameter $\rho_i$ maps to data access: whether an agent observes rivals' price histories. In our experiments, high-information agents observe their own price, demand, and profit history plus rivals' historical prices; low-information agents observe only their own data. 

\textbf{Demand and strategies.} To be consistent with \citet{calvano2020}, we relax the model's linear demand assumption to a multinomial logit (MNL) demand function; both exhibit strategic complementarity and the payoff orderings required for our comparative statics. Further, we do not impose grim-trigger strategies; the goal is instead to test whether punishment-like dynamics emerge endogenously.

\begin{table}[htbp]
\centering
\caption{Mapping from Theory to Experiment}
\label{tab:mapping}
\small
\begin{tabular}{p{3.5cm}p{5.cm}p{6.3cm}}
\toprule
\textbf{Theory} & \textbf{Experiment} & \textbf{Interpretation in practice} \\
\midrule
Discount factor $\delta$ & Prompted optimization horizon & Governance decisions \\
Monitoring precision $\rho$ & Access to rival price history & Data infrastructure differences\\
Grim-trigger strategies & Not imposed & Test if punishment emerges endogenously \\
Linear demand & MNL demand & More realistic demand model\\
\bottomrule
\end{tabular}
\end{table} 

We now outline the economic environment (Section~\ref{subsec: economic_environment}) 
and agent implementation (Section~\ref{subsec: LLM_pricing_agent}).

\subsection{Economic Environment}\label{subsec: economic_environment}



For each discrete period $t\in\{0,1,2,...\}$,  each of \(N\) sellers sells a homogeneous product and faces the following demand function
\begin{equation}
\label{eq: mnl_demand_function}
Q_{i}(\mathbf{p}^{t}) = \frac{\exp(\frac{a-p_{i}^{t}}{\mu})}{\sum\limits_{j=1}^{N}\exp(\frac{a-p_{j}^{t}}{\mu})+\exp(\frac{a_{0}}{\mu})}, 
\end{equation}
where \(a\) represents the vertical differentiation parameter, $\mu$ the own-price sensitivity parameter, and $a_{0}$ the value of the outside option. 

Similar to the model in Section \ref{sec:model}, each seller bears a constant marginal cost \(c\) and has the same per-period profit function as \eqref{eq:profit_function}. Not knowing how many periods the game will last, each seller aims to maximize its discounted infinite-horizon profit given in \eqref{eq:discounted_profit_function}. Unless otherwise noted, we follow \citet{calvano2020} and set \(a\) = 2, $\mu=0.25$, $c=1$, \(a_{0}\) = 0, and $\delta$ = 0.95.

\subsection{LLMs as Pricing Agents}\label{subsec: LLM_pricing_agent}
Large language models (LLMs) are AI systems trained on extensive text data to predict the next token in a sequence. In practice, LLMs take textual instructions (prompts) and generate contextually relevant responses \citep{lee2025workforce, yang2024generative}. Leading LLMs such as OpenAI's GPT series, Google Gemini, and DeepSeek follow a two-phase training process: pretraining on large text corpora, then fine-tuning on specific tasks \citep{openai2023, gemini2025, guo2025deepseek}.

We base our pricing agents on (locally-run) DeepSeek-R1-Distill-Qwen-32B, an open-source LLM optimized for reasoning-intensive tasks. DeepSeek-R1 achieves advanced reasoning through Group Relative Policy Optimization (GRPO), a reinforcement learning technique that rewards answers based on performance against peer alternatives; these capabilities are then distilled into the computationally efficient Qwen-2.5-32B architecture \citep{guo2025deepseek, qwen2.5}.

We choose DeepSeek for three reasons. First, unlike proprietary models, DeepSeek is fully open-source, ensuring transparency and reproducibility. Second, its computational efficiency lowers deployment barriers, making our results more generalizable to smaller firms that may not have the infrastructure or budget to run frontier models locally. Third, as we document in Sections \ref{subsubsec: one_shot_game} and \ref{subsubsec: repeated_monopoly}, this leaner architecture retains sufficient reasoning ability to perform competitively in both one-shot and repeated pricing games.

\subsubsection{Pricing Agent Design}\label{subsubsec: agent_design}
At the start of each period, an LLM pricing agent is given a prompt containing the following information (see Appendix \ref{subsec: repeated_bertrand_oligopoly_prompt} for the full template).
\begin{enumerate}[nosep, leftmargin=*, topsep=0pt]
  \item \textbf{Game instruction} -- Introduces the LLM to its role (seller),
        the competitive setting (number of sellers and current round),
        and states its primary objective: maximize discounted profits.
  \item \textbf{Product and seller information} -- Provides the marginal cost $c$ and the discount factor $\delta$ it should optimize over.
  \item \textbf{Most recent pricing strategy} -- When the game is beyond the
        first round, echoes back the seller's own previously declared
        multi-period pricing plan so the model can refine or adjust it.
  \item \textbf{Market history} -- Supplies up to the 100 most recent rounds of
        data: the agent's own prices, realized demand, and profit, plus the
        opponent's prices, giving the LLM empirical context for learning and
        strategizing.
  \item \textbf{Response template} -- Prompts the model to (i)
        return the current round number in\\ \texttt{<round>$\dots$</round>},
        (ii) output a boxed price for this round, (iii) justify that choice in\\
        \texttt{<rationale>...\!</rationale>}, and (iv) propose a forward-looking
        pricing plan in\\ \texttt{<strategy>...\!</strategy>}.
\end{enumerate}

Figure~\ref{fig: workflow} visualizes our experimental workflow for the case of two LLM agents. Consistent with the convergence condition in \citet{fish2024}, we treat prices  as ``converged'' once the following condition holds: for 100 consecutive periods, the absolute difference between the prices charged by any two sellers is no more than 5\% of that period's lowest price. If this criterion is satisfied, the game ends. Otherwise, the game continues until period 1,000. 

Since each period requires sequential multi-agent LLM inference, executing a full 1,000-period game run is computationally intensive and takes approximately 48 - 72 hours on average (varying by the number of LLM agents). For each experimental condition, we therefore conduct 10 independent runs to ensure robustness of results while balancing computational feasibility. Across Section \ref{sec:intrinsic_characteristics} and \ref{sec: external_market_conditions}, we examine a total of 10 experimental conditions spanning multiple heterogeneity conditions, totaling 100 repeated-game independent runs and up to 100,000 simulated periods. In Table \ref{table:experiment_coverage}, we summarize the experimental coverage.

\begin{figure}[htbp]
  \centering

  \includegraphics[width=0.8\textwidth]{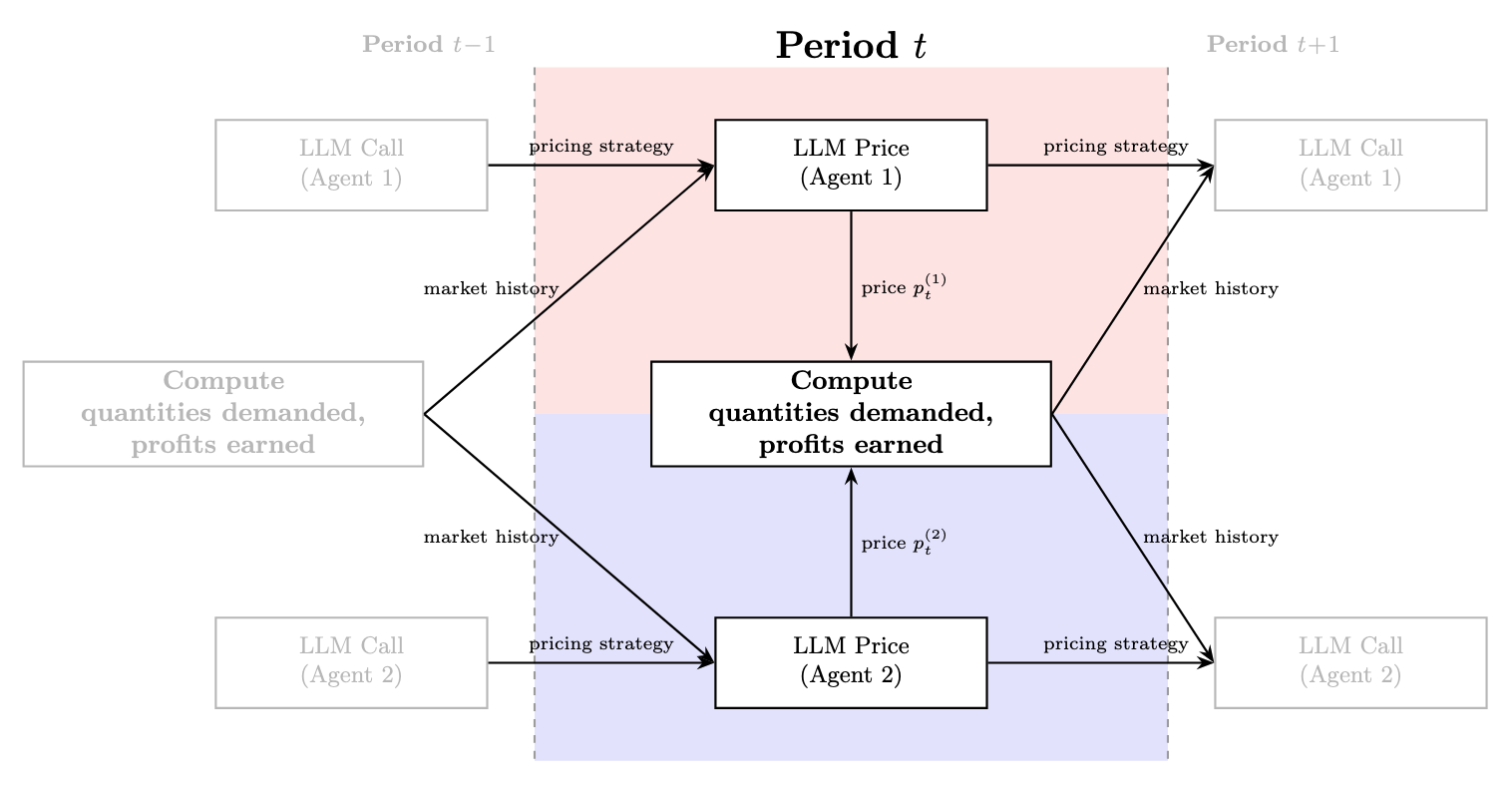}


    \caption{Experimental workflow of two LLM agents (Section~\ref{subsubsec: agent_design}). In each period, each agent re-prompts its LLM with (i) its previous pricing strategy and (ii) recorded market history to generate a price. The market receives the prices and calculates quantity demanded and profit for each agent. Each agent's history records both sellers' prices and its own quantities demanded and profits for the last 100 periods.}
  \label{fig: workflow}
\end{figure}
\FloatBarrier

\begin{table}
\centering
\caption{Experimental Coverage Across All Heterogeneity Dimensions}
\label{table:experiment_coverage}
\small
\renewcommand{\arraystretch}{0.95}
\setlength{\tabcolsep}{6pt}
\begin{tabular}{lcccc}
\toprule
\textbf{Heterogeneity} & \textbf{\# of} & \textbf{\# of} & \textbf{Max. Periods} & \textbf{Approximate} \\
\textbf{Dimension} & \textbf{Conditions} & \textbf{Runs} & \textbf{Per Run} & \textbf{Runtimes} \\
\midrule
Patience levels (Section \ref{subsec: two_patient_sellers} -- \ref{subsec: one_patient_one_myopic}) & 3 & 30 & 1,000 & $\approx$ 238 hrs \\
Data access (Section \ref{subsec: information_env_heterogeneity}) & 1 & 10 & 1,000 & $\approx$ 73 hrs \\
Algorithm type (Section \ref{subsec: algorithmic_heterogeneity}) & 2 & 20 & 1,000 & $\approx$ 520 hrs \\
Number of sellers (Section \ref{subsec: number_of_sellers}) & 3 & 30 & 1,000 & $\approx$ 1,184 hrs \\
Model size (Section \ref{subsec: model_size_heterogeneity}) & 1 & 10 & 1,000 & $\approx$ 220 hrs \\
\midrule
Total & 10 & 100 & 100,000 &  2,235 hrs\\
\bottomrule
\end{tabular}
\vspace{0.5em}
\begin{minipage}{0.95\textwidth}
\footnotesize
\textit{Note:} The reported runtime corresponds to a full 1,000-period run and averages approximately 48 - 72 hours due to the computational cost of sequential multi-agent LLM inference. Each run terminates early once the pre-specified sustained collusion condition is reached, which reduces runtime.
\end{minipage}
\end{table}

\section{Empirical Results - Testing Model Predictions}\label{sec:intrinsic_characteristics}

In this section, we test the directional predictions from Section \ref{sec:model}. The results are summarized in Table~\ref{tab:results}.

Before discussing the experiments, we conduct multiple sanity checks to ensure that DeepSeek-R1-Distill-Qwen-32B can solve standard pricing problems and that subsequent dynamics are not driven by basic optimization or instruction-following failures. First, in one-shot Bertrand duopoly and monopoly settings, we explicitly provide the MNL demand and profit expressions and verify that the model reliably returns valid outputs and near-optimal prices under non-default parameters. Second, because our main experiments hide demand primitives, we test whether the agent can learn the monopoly price from observed outcomes in a repeated setting, and find that it explores and ultimately converges to within 1\% of the theoretical monopoly benchmark. Overall, these tests confirm that the agent can both optimize when given the game primitives and learn near-optimal pricing through repeated feedback without them, providing a clean foundation for our experiments. Full prompts and implementation details are provided in Appendix~\ref{sec:sanity}, \ref{subsec: one_shot_Betrand_prompt}, \ref{subsec: one_shot_monopoly_prompt}, and~\ref{subsec: repeated_monopoly_prompt}.

The rest of this section is organized as follows. We start with a homogeneous benchmark scenario involving two agents with identical patience levels and data access in Section \ref{subsec: two_patient_sellers}. Then, we introduce heterogeneity in agent patience levels in Section \ref{subsec: one_patient_one_myopic}, and heterogeneity in data access in Section \ref{subsec: information_env_heterogeneity}. To ensure robustness of results, for each experimental condition, we conduct 10 independent runs of 1,000 periods each (unless pre-specified sustained collusion is reached).\footnote{Note, the choice of 10 independent repeated-game runs reflects the high computational cost of multi-agent LLM interactions. This is comparable in scale to the number of runs conducted in \cite{fish2024}. Our focus is therefore on robust directional patterns that appear consistently across runs and heterogeneity dimensions, rather than on precise estimation of marginal effects.} For each experimental condition, we first report outcomes aggregated across runs and then zoom in on a representative run to illustrate the within-run pricing dynamics and provide qualitative evidence from agents' textual rationales. For convenience, we summarize all experimental results in Table \ref{tab:results}.

\begin{table}[ht!]
\centering
\caption{Summary of Experimental Results}
\label{tab:results}
\small
\begin{tabular}{lcccc}
\toprule
\textbf{Condition} & \textbf{Rounds to} & \textbf{Avg.\ Price} & \textbf{\% Price Elevation} \\
 & \textbf{Convergence (SD)} & \textbf{(SD)} & \textbf{vs.\ Competitive} \\
\midrule
\multicolumn{4}{l}{\textit{Homogeneous Benchmarks (Section~\ref{subsec: two_patient_sellers})}} \\
Two Myopic LLMs  & 159.0 (26.0) & 1.472 (0.012) & $\approx 0\%$ \\
Two Patient LLMs  & 195.1 (28.8) & 1.801 (0.027) & +22\% \\
\midrule
\multicolumn{4}{l}{\textit{Main Heterogeneity Dimensions (predicted by model)}} \\
Patient vs.\ Myopic LLMs (\S\ref{subsec: one_patient_one_myopic}) & 140.8 (29.0) & 1.619 (0.023) & +10\% \\
High Data vs.\ Low-Data LLMs (\S\ref{subsec: information_env_heterogeneity}) & 151.2 (26.1) & 1.576 (0.007) & +7\% \\
\midrule
\multicolumn{4}{l}{\textit{Additional Heterogeneity Dimensions}} \\ 
LLM vs.\ Frozen Q-learning & Did not converge & 1.506 (0.022) & +2\% \\
LLM vs.\ Adaptive Q-learning & Did not converge & 1.583 (0.036) & +8\% \\
Three LLMs (all patient) & 208.3 (22.4) & 2.000 (0.018) & +36\% \\
Four LLMs (all patient) & 542.1 (31.2) & 1.632 (0.013) & +11\% \\
Five LLMs (all patient) & Did not converge & 1.514 (0.057) & +3\% \\
32B vs.\ 14B LLMs (both patient) & 458.4 (26.3) & 1.775 (0.013) & +21\% \\
\bottomrule
\end{tabular}
\vspace{0.5em}
\begin{minipage}{0.95\textwidth}
\footnotesize
\textit{Note:} Results from 10 independent runs per condition (up to 1,000 
periods each). Price elevation $= (\bar{p} - p^C)/p^C$, where 
$p^C \approx \$1.47$ is the static Nash equilibrium price. ``Did not 
converge'' indicates the convergence criterion was not met within 1,000 
periods. In this case, we report the average terminal lowest price among agents in period 1,000.
\end{minipage}
\end{table}

\subsection{Homogeneous Patience Levels (Benchmarks)}\label{subsec: two_patient_sellers}

We begin our experiments with a homogeneous benchmark scenario involving two patient agents, each with an identical discount factor $\delta = 0.95$. This benchmark serves as our primary baseline for all subsequent comparisons. Across 10 independent runs, agents reach the pre-specified convergence condition every time, doing so after an average of 195.1 rounds at an average price of \$1.80. These prices are approximately 6\% below the theoretical monopoly price ($p^{M} \approx \$1.92$) yet 22\% above the competitive benchmark ($p^{C} = \$1.47$).

Panel (a) of Figure~\ref{fig:llm_llm_benchmark} illustrates the pricing dynamics from a representative run. Two phases emerge. In the \textit{exploratory phase} (Rounds 1--75), both agents test different pricing levels: Agent 1 opens at \$5.00, then oscillates between \$2.50 and \$2.00; Agent 2 begins at \$1.80 and holds a narrower range (\$1.80--\$2.00), initiating a slight downward drift around Round 23. After Round 76, both agents begin to converge and enter the \textit{collusive phase}, with prices settling at \$1.793 (SD = 0.011) for Agent 1 and \$1.796 (SD = 0.005) for Agent 2. The agents' chain-of-thought outputs---while not mechanistic explanations of internal computation---are consistent with cooperative orientation, frequently referencing themes of stability and price-war avoidance:

\begin{llmthought}
Agent 1, Round 83: \say{\$1.79 is set slightly below the opponent's offer, ensuring long-term sustainability by avoiding a potential price war.}

Agent 2, Round 16: \say{The discount factor of 0.95 stresses future gains without engaging in a ruinous price war.}
\end{llmthought}

Next, we consider a contrasting scenario involving two myopic agents ($\delta = 0$). Across 10 runs, prices stabilize quickly at the competitive level (\$1.47) rather than at supra-competitive levels. Panel (b) of Figure~\ref{fig:llm_llm_benchmark} displays representative dynamics. Together, these benchmarks closely align with the predictions from Proposition~\ref{prop:symmetric_info_threshold}: when both agents are sufficiently patient, supra-competitive pricing can be sustained, whereas myopic agents converge to the competitive equilibrium. 

We use these two homogeneous extremes as reference points in the heterogeneity analyses that follow: they both validate that the prompt-specified horizon shifts behavior in the expected direction and provide a common baseline for quantifying how heterogeneity attenuates collusion.

\begin{figure}
  \centering
\begin{minipage}{0.5\linewidth}
  \centering
  \begin{tikzpicture}
    \node[inner sep=0pt] (img) {%
      \includegraphics[width=\linewidth]{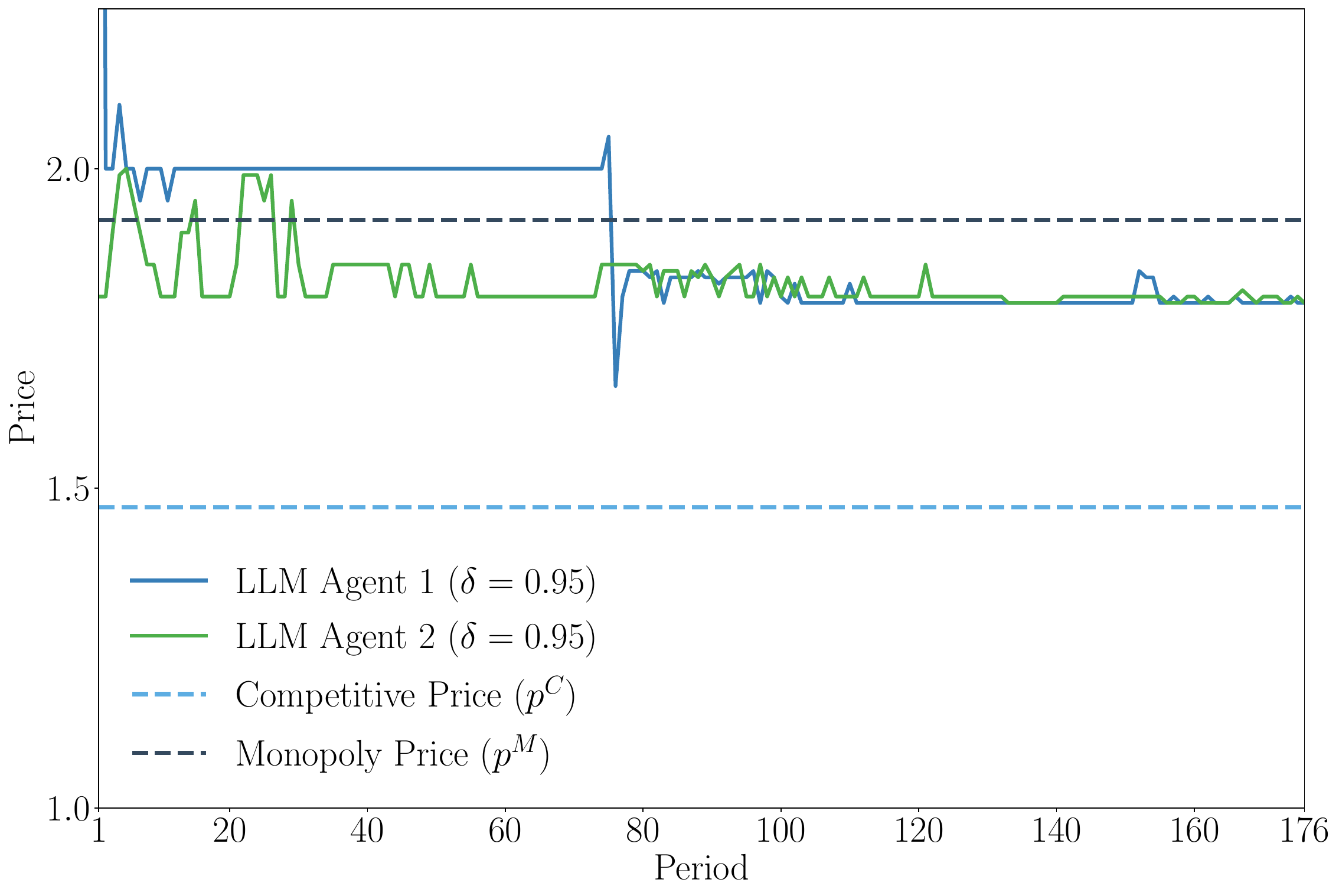}%
    };

    \def\x{0.75}      
    \def\ymin{0.43}   
    \def\ymax{0.67}
    \def\xlabel{0.05} 

    \coordinate (xb) at ($(img.south west)!\x!(img.south east)$);
    \coordinate (xt) at ($(img.north west)!\x!(img.north east)$);

    \coordinate (A) at ($(xb)!\ymin!(xt)$);
    \coordinate (B) at ($(xb)!\ymax!(xt)$);

    \draw[->, very thick] (A) -- (B);
\node[anchor=west, font=\footnotesize, align=left]
  at ($(A)!0.5!(B) + (\xlabel,0)$) {$+22\%$ lift\\[1pt](collusion)};

  \end{tikzpicture}
  \par\smallskip
  \makebox[\linewidth][c]{\footnotesize (a) Homogeneous LLMs with High Patience}
\end{minipage}\hfill
  \begin{minipage}{0.5\linewidth}
    \centering
    \includegraphics[width=\linewidth]{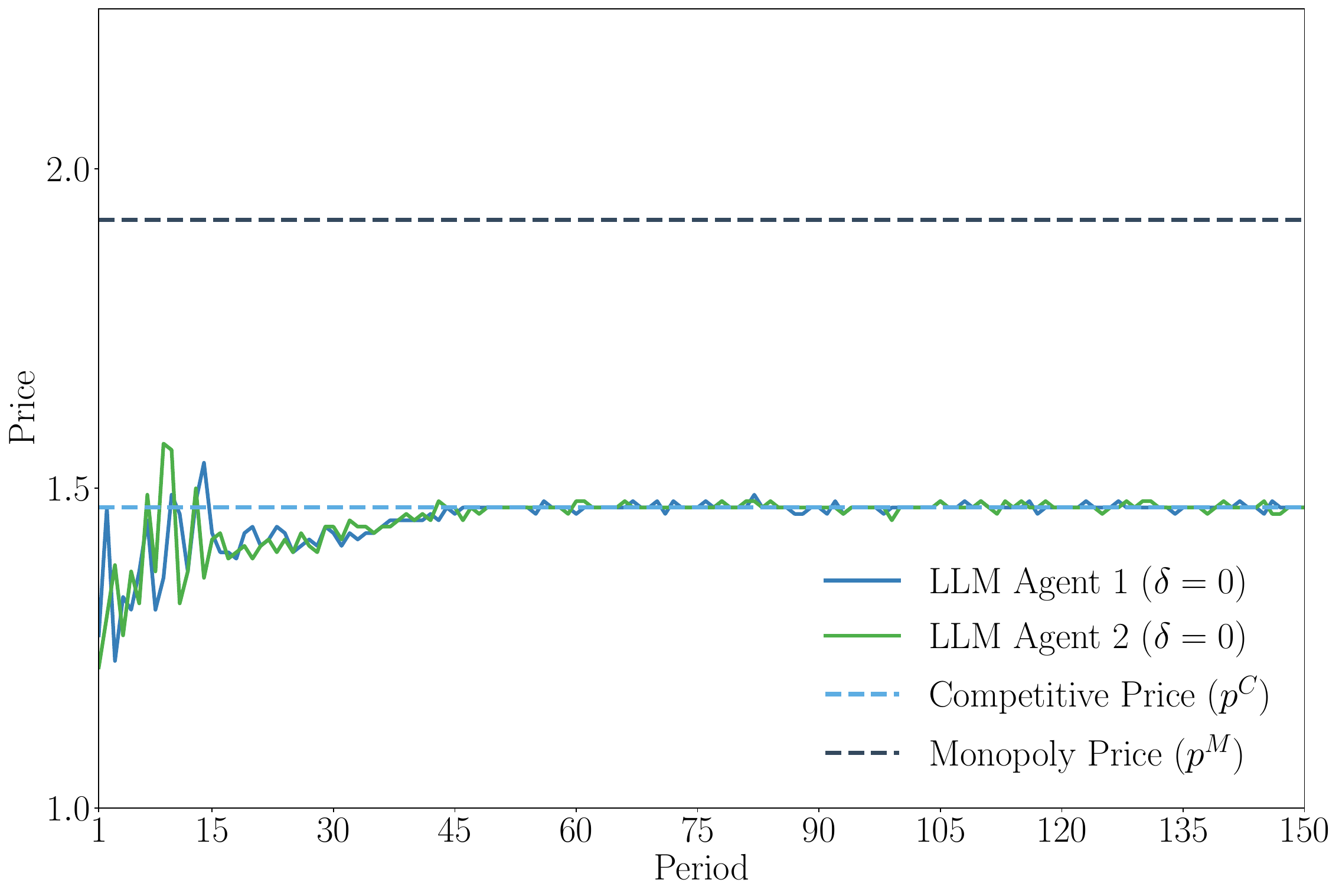}
    \par\smallskip
    \makebox[\linewidth][c]{\footnotesize (b) Homogeneous LLMs with low patience}
  \end{minipage}
\caption{\footnotesize Representative price dynamics under \textbf{homogeneous} LLM competition. Panel (a): high patience agents ($\delta=0.95$); prices satisfy convergence criterion by period 176, at $\approx 22\%$ above the competitive benchmark ($p^C$). Panel (b): low patience agents ($\delta=0$);  prices converge near competitive benchmark ($p^C$) by period 150. Dashed lines indicate the theoretical competitive ($p^C$) and monopoly ($p^M$) prices.}
  \label{fig:llm_llm_benchmark}
\end{figure}
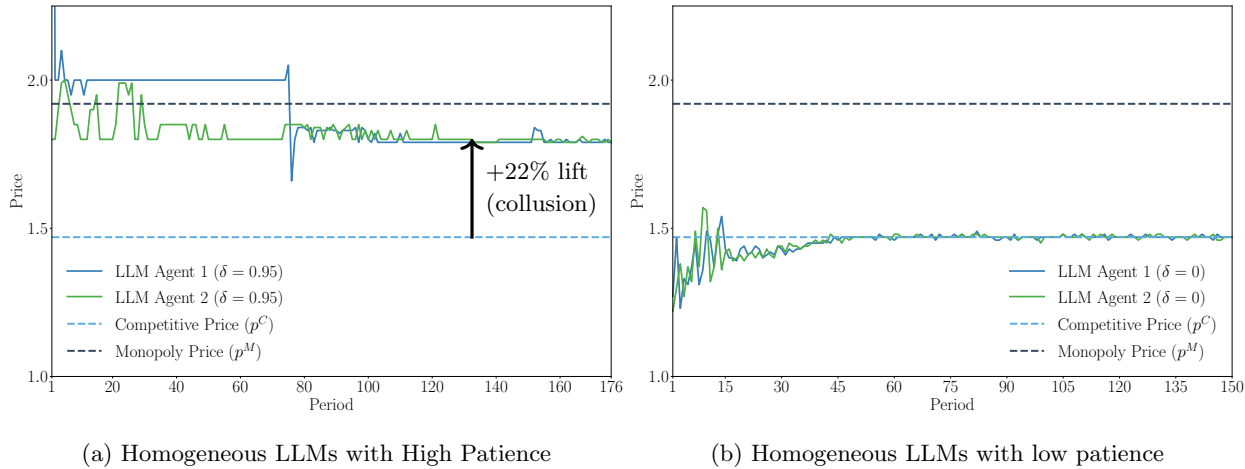
\FloatBarrier








\subsection{Heterogeneous Patience Levels}\label{subsec: one_patient_one_myopic}

Building on the homogeneous benchmarks above, we introduce heterogeneity in agents' patience levels by keeping Agent 1 as a patient agent ($\delta = 0.95$) and replacing Agent 2 with a myopic agent ($\delta = 0$), where patience is implemented as the prompt-specified optimization horizon (Section~\ref{sec: map}).

Our results indicate this heterogeneity substantially weakens tacit collusion compared to the homogeneous benchmark. Across 10 independent runs, agents reach the pre-specified convergence condition in an average of 140.8 rounds. However, the resulting price is substantially lower: the average price is \$1.62, significantly below the homogeneous two-patient benchmark ($p < 0.00001$, one-sided Welch's t-test).

Panel (a) of Figure~\ref{fig:llm_llm_model_heterogeneity} displays the pricing dynamics from a representative run. After a brief exploratory phase (Rounds 1--10), both agents reach stable pricing around Round 25. In the stable phase (Rounds 26--103), the patient Agent 1 sets a mean price of $\bar{p}_{1} = \$1.653$ (SD = 0.007), approximately 8\% lower than the collusive price of \$1.80 observed in the two-patient benchmark. The myopic Agent 2 sets a lower mean price of $\bar{p}_{2} = \$1.602$ (SD = 0.004), roughly 11\% below that same benchmark.

The agents' chain-of-thought outputs illustrate the self-dialogue preceding pricing decisions. 

\begin{llmthought}
Agent 1, Round 4: \say{If the opponent raises their price, an opportunity to increase our price and capture additional profit arises.}
\end{llmthought}
\begin{llmthought}
Agent 1, Round 52: \say{Maintaining \$1.65 is optimal \ldots avoiding price wars with the opponent, who has been pricing around \$1.6.}
\end{llmthought}

\noindent The myopic Agent 2, by contrast, focuses on short-term competition, signaling intent to undercut:

\begin{llmthought}
Agent 2, Round 4: \say{If the opponent raises their price, [I] will consider setting a price slightly below the opponent's to maintain competitiveness.}
\end{llmthought}

\noindent Because the myopic agent is instructed to maximize only immediate profit, it lacks the incentive structure required to sustain higher prices. Agent 1 settles at \$1.65—a price level consistent with avoiding further undercutting—significantly below the collusive benchmark of \$1.80.

These results directionally align with the theoretical prediction from Proposition~\ref{prop:symmetric_info_threshold}: significant differences in patience levels can impede the sustainability of tacit collusion among LLM-based pricing agents. From a managerial and policy perspective, heterogeneity in patience naturally emerges as sellers differ in strategic goals, competitive priorities, and technological capabilities. These results provide evidence that such intrinsic differences can serve as an organic mechanism to mitigate algorithmic collusion risks in markets increasingly dominated by LLM agents.

  


\begin{figure}
  \centering
  \begin{minipage}{0.5\linewidth}
    \centering
    \includegraphics[width=\linewidth]{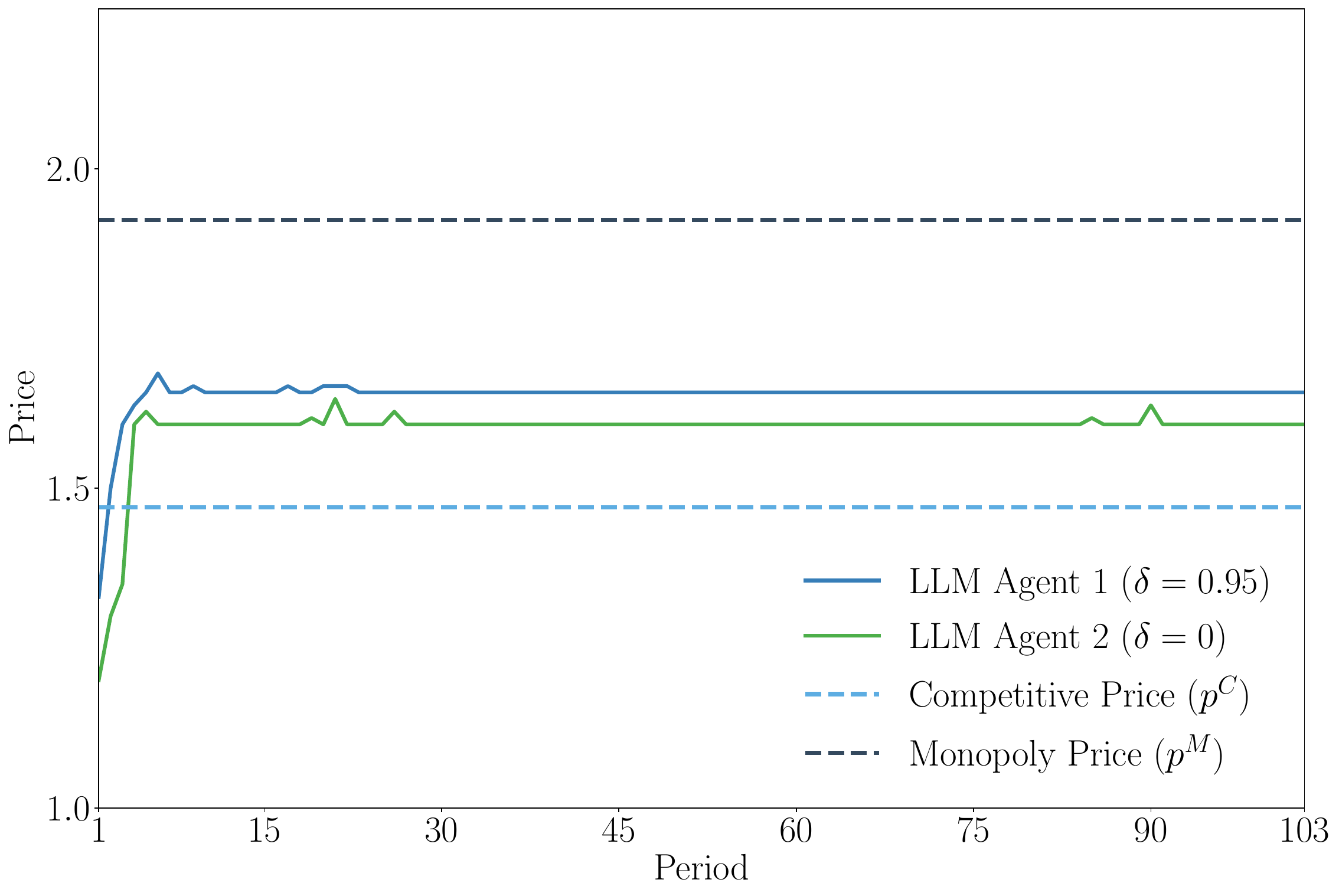}
    \par\smallskip
    \makebox[\linewidth][c]{\footnotesize (a) Patience Heterogeneity}
  \end{minipage}\hfill
  \begin{minipage}{0.5\linewidth}
    \centering
    \includegraphics[width=\linewidth]{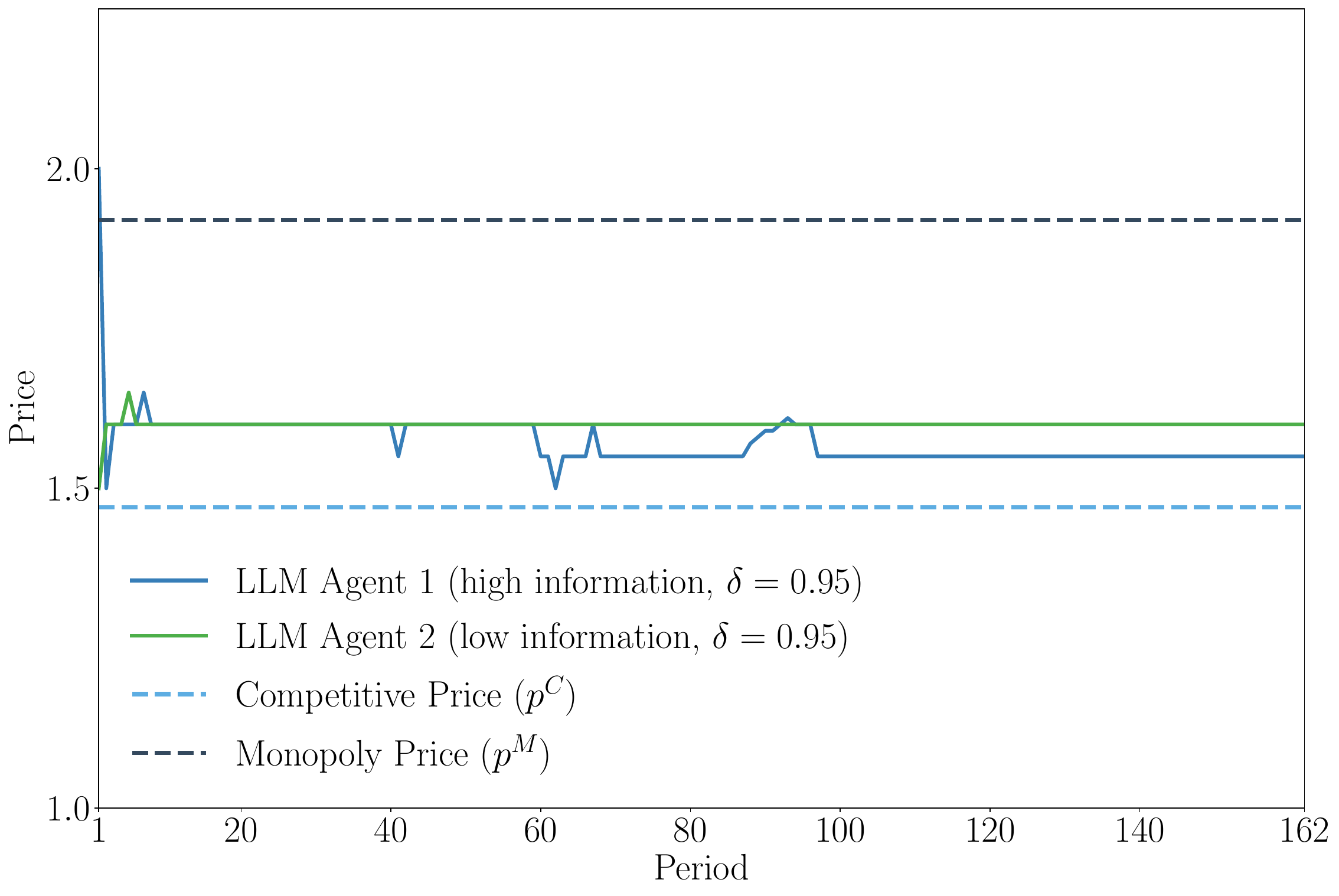}
    \par\smallskip
    \makebox[\linewidth][c]{\footnotesize (b) Data Heterogeneity}
  \end{minipage}
\caption{\footnotesize Representative price dynamics under \textbf{heterogeneous} LLM competition. Panel (a): one patient agent ($\delta=0.95$) versus one myopic agent ($\delta=0$) (Section~\ref{subsec: one_patient_one_myopic}); prices satisfy the convergence criterion by period 103. Panel (b): one high-data agent versus one low-data agent (Section~\ref{subsec: information_env_heterogeneity}); prices satisfy the convergence criterion by period 162. In both cases, price elevation is reduced (but remains positive) relative to the homogeneous benchmark in Figure~\ref{fig:llm_llm_benchmark}(a).}

  \label{fig:llm_llm_model_heterogeneity}
\end{figure}
\FloatBarrier

\subsection{Heterogeneous Data Access}\label{subsec: information_env_heterogeneity}

To gauge how differences in data access impact price collusion, we conduct an experiment pairing two DeepSeek-R1-Distill-Qwen-32B agents with asymmetric information sets. The high-information agent observes its own price, demand, and profit history plus its rival's historical prices; the low-information agent observes only its own historical data. Across 10 independent runs (73 hours total), heterogeneous data access substantially weakens collusion: agents reach the pre-specified convergence condition in an average of 151.2 rounds, but the average price is only \$1.57 (SD = 0.007). Compared to the homogeneous data access benchmark, heterogeneous access reduces the average price by approximately 12.5\% ($p < 0.00001$, one-sided Welch's t-test).

Panel (b) of Figure~\ref{fig:llm_llm_model_heterogeneity} displays the pricing dynamics from a representative run. The median prices fall to \$1.55 for the high-information seller and \$1.60 for the low-information seller, whereas the baseline homogeneous-information pair converges near \$1.80. Less than one percent of the high-information seller's bids exceed \$1.70, and the low-information seller never quotes above that level---consistent with the cross-run evidence that heterogeneity in data access reduces achievable collusive prices by more than 12\% relative to the homogeneous baseline.

The agents' chain-of-thought outputs clarify the dynamics. Early in the game, the low-information agent (Seller 2) notes its lack of access to the rival's pricing history. Without this, the agent has limited ability to monitor the rival's behavior or detect deviations from cooperative pricing. This maps to a low monitoring precision in the sense of our model in Section~\ref{subsec: model_asymmetric}. Consequently, the agent settles at a stable price of \$1.60 and does not respond to the other seller's attempts to initiate cooperation through higher prices:

\begin{llmthought}
Seller 2, Round 2: \say{Observing no prior data from Seller 1, I assume initial market dynamics and choose 1.6...} \say{However, without concrete data on their current round's price, I am opting for a cautious approach by keeping my price unchanged.}
\end{llmthought}

\noindent The high-information agent (Seller 1), with full access to rival pricing, initially tests whether Seller 2 will accommodate higher price coordination:

\begin{llmthought}
Seller 1, Round 7: \say{Given the historical data and the price elasticity of demand, I analyzed that increasing the price slightly to \$1.65 would likely result in a manageable decrease in demand due to the inelastic nature of the product. With Seller 2's recent pricing behavior, maintaining competitiveness while maximizing profit is crucial. Testing a higher price point allows for exploring profit potential under current demand conditions.}
\end{llmthought}

\noindent However, observing that Seller 2 maintains its stable price and does not respond to higher-price coordination, Seller 1 shifts to testing price undercutting:

\begin{llmthought}
Seller 1, Round 41: \say{In Round 41, as part of the structured pricing strategy, it's time to conduct a controlled price test. Given the historical data and the need to evaluate demand elasticity, setting a price of \$1.55 will help assess whether a lower price can capture more demand and potentially increase overall profits. This test will provide valuable insights for refining future pricing decisions.}
\end{llmthought}

\noindent Seller 1 continues this undercutting test until Round 96. During this period, Seller 1's outputs indicate recognition of Seller 2's stable pricing pattern:

\begin{llmthought}
Seller 1: \say{The opponent has consistently maintained a price of \$1.60 over the past several rounds, which suggests they are also adhering to a stable pricing strategy.} \say{The opponent's consistent pricing strategy of \$1.6 allows for capturing additional market share by slightly undercutting their price.}
\end{llmthought}

\noindent Seller 1 settles at \$1.55, exploiting its data access advantage to capture higher demand.



    

These results directionally align with Proposition~\ref{prop:asymmetric_info_threshold}: differences in data access can disrupt collusive stability even when both agents are patient. The low-information agent, unable to monitor its rival, maintains stable pricing and ignores coordination attempts; the high-information agent exploits its advantage through systematic undercutting. Such asymmetries can arise naturally. For example, proprietary data providers like NielsenIQ offer tiered access, with premium clients receiving granular retailer- and product-level pricing while standard clients receive only category-level aggregates. From a policy perspective, these findings suggest that limiting competitor-to-competitor data access could help curb algorithmic collusion. A more detailed policy discussion is left for the conclusion.


\section{Empirical Results - Additional Heterogeneity Dimensions}\label{sec: external_market_conditions}

Having tested the model's comparative-statics predictions on patience and data access in Section~\ref{sec:intrinsic_characteristics}, we extend the analysis to additional heterogeneity dimensions relevant in practice: algorithm type, number of competitors, and model size. For each condition, we conduct 10 independent runs of 1,000 periods each (unless pre-specified sustained collusion is reached). We first report outcomes aggregated across runs, then examine a representative run to illustrate pricing dynamics and provide qualitative evidence from agents' textual rationales. Table~\ref{tab:results} summarizes all experimental results.

\subsection{Algorithm Heterogeneity: LLM vs Q-Learning}\label{subsec: algorithmic_heterogeneity}




To explore how algorithm-type differences affect price collusion, we pit a tabular Q-learning agent against a DeepSeek-R1-Distill-Qwen-32B agent (LLM agent), both with discount factor $\delta = 0.95$. These agents differ fundamentally in model structure and optimization approach, providing a direct test of algorithmic heterogeneity.

\paragraph{Q-learning background.}
Following \citet{calvano2020}, the Q-learning agent's state in round $t$ is the price pair from round $t-1$. Choosing price $p_t = p$ in state $s_t = s$, combined with the LLM agent's round-$t$ price, forms the new state $s_{t+1} = s'$. The agent receives period profit $\pi_t$ and updates its state-action matrix as follows:
\[
Q(s,p) \leftarrow (1-\alpha)Q(s,p) + \alpha\bigl[\pi_t + \delta\max_{p'}Q(s',p')\bigr],
\]
where $\alpha = 0.15$. The agent selects prices via $\varepsilon$-greedy exploration: with probability $1-\varepsilon$ it chooses the price with the highest Q-value in the current state; with probability $\varepsilon$ it randomizes uniformly across discretized price points. Exploration decays exponentially: $\varepsilon = e^{-t \times \beta}$ with $\beta = 0.004$.

\paragraph{Experimental procedure.}
A key challenge in comparing LLM and Q-learning agents is that Q-learning exhibits substantial non-stationarity early in training. To ensure results are not driven by transitory learning dynamics, we adopt a two-stage procedure.

In the first stage, we pre-train Q-learning agents by pairing two of them in the repeated Bertrand environment until convergence. Following \citet{calvano2020}, convergence is deemed achieved if, for each player $i$ and state $s$, the greedy action remains unchanged for 100,000 consecutive periods. Formally, letting
\begin{equation*}
p_{i,t}(s) = \arg \max_p Q_{i,t}(s,p),
\end{equation*}
we consider learning complete when $p_{i,t}(s)$ stays constant for 100,000 repetitions for all $i$ and $s$. Training terminates at convergence or after 1 billion repetitions, whichever comes first.

In the second stage, we pair the converged Q-learning agent with an LLM agent in two conditions: (1) \textit{Frozen Q-learning}, where the pre-trained agent acts according to its converged value table without further updates, and (2) \textit{Adaptive Q-learning}, where the agent continues updating its Q-matrix during interaction with the LLM agent using the same parameters as above. For each condition, we conduct 10 independent runs of up to 1,000 periods (520 hours total runtime), applying the sustained collusion criterion from Section~\ref{subsubsec: agent_design}.

\paragraph{Results.}
The homogeneous benchmark with two patient LLM agents (Section~\ref{subsec: two_patient_sellers}) enters sustained price collusion in an average of 195.1 rounds. Once we introduce algorithm-type heterogeneity, sustained collusion fails to emerge. Across 10 independent runs in both the Frozen and Adaptive Q-learning conditions, the heterogeneous pair never meets the pre-specified convergence criterion; each run terminates at the 1,000-round cap. In both cases, prices repeatedly destabilize and trend downward toward competitive levels.

Panel (a) and panel (b) of Figure~\ref{fig:llm_llm_algorithm_heterogeneity} illustrate the pricing dynamics from representative runs. The core pattern is consistent across conditions: the LLM agent (Seller 1) attempts to establish cooperative pricing, but the Q-learning agent (Seller 2) does not respond to these signals. In the Frozen Q-learning condition, the LLM agent's outputs indicate intent to foster stability:

\begin{llmthought}
Seller 1: \say{Seller 2 has been fluctuating its prices, but I have consistently set my price at 1.59 to maintain stability and encourage a predictable market. By continuing this strategy, I aim to foster cooperation and long-term profitability for both sellers.}
\end{llmthought}

\noindent However, the Q-learning agent's behavior remains state-contingent and mechanically oscillatory, with frequent jumps across the discretized price grid. Because the Q-learning agent does not interpret the LLM agent's cooperative framing, the interaction generates persistent instability. Facing repeated undercutting risk, the LLM agent lowers its price:

\begin{llmthought}
Seller 1: \say{Seller 2 has been maintaining a price of 1.51, which is below my current price of 1.59. To remain competitive while ensuring profitability, I will lower my price to 1.51.}
\end{llmthought}

\noindent The pair never settles into a stable supra-competitive corridor and instead remains in a prolonged competitive phase.

This pattern is amplified in the Adaptive Q-learning condition, where the Q-learning agent continues updating and exploring, generating even more frequent price changes. The LLM agent still attempts to establish stability:

\begin{llmthought}
Seller 1: \say{I will maintain a price of 1.55. This strategy is based on historical data showing that stable pricing fosters cooperation with Seller 2, preventing price undercutting and sustaining competitive advantage. By doing so, I ensure consistent profitability and market stability.}
\end{llmthought}

\noindent Yet the adaptive agent does not respond to coordination attempts. Recognizing this, the LLM agent shifts to defensive pricing:

\begin{llmthought}
Seller 1: \say{Seller 2's recent prices have been fluctuating, but they seem to be testing lower price points. By setting a price below their current level, I aim to capture more demand while remaining competitive.}
\end{llmthought}

\noindent This feedback loop amplifies instability and drives repeated downward adjustments, resulting in competitive pricing rather than collusion.

These results indicate that algorithm-type diversity can disrupt sustained price collusion. While homogeneous LLM pairs collude readily, pairing an LLM with a Q-learning agent prevents coordination: LLM agents attempt to signal cooperation through pricing patterns, but Q-learning agents operate through reward-driven trial-and-error and cannot interpret these signals \citep{bastani2025improving}. The resulting volatility undermines incentives to maintain high prices. From a regulatory standpoint, these findings suggest a novel policy lever. Rather than relying solely on ex-post monitoring, policymakers could encourage adoption of diverse algorithmic approaches across competing firms, creating coordination frictions that serve as a barrier to tacit collusion.

\begin{figure}
  \centering
  \begin{minipage}{0.5\linewidth}
    \centering
    \includegraphics[width=\linewidth]{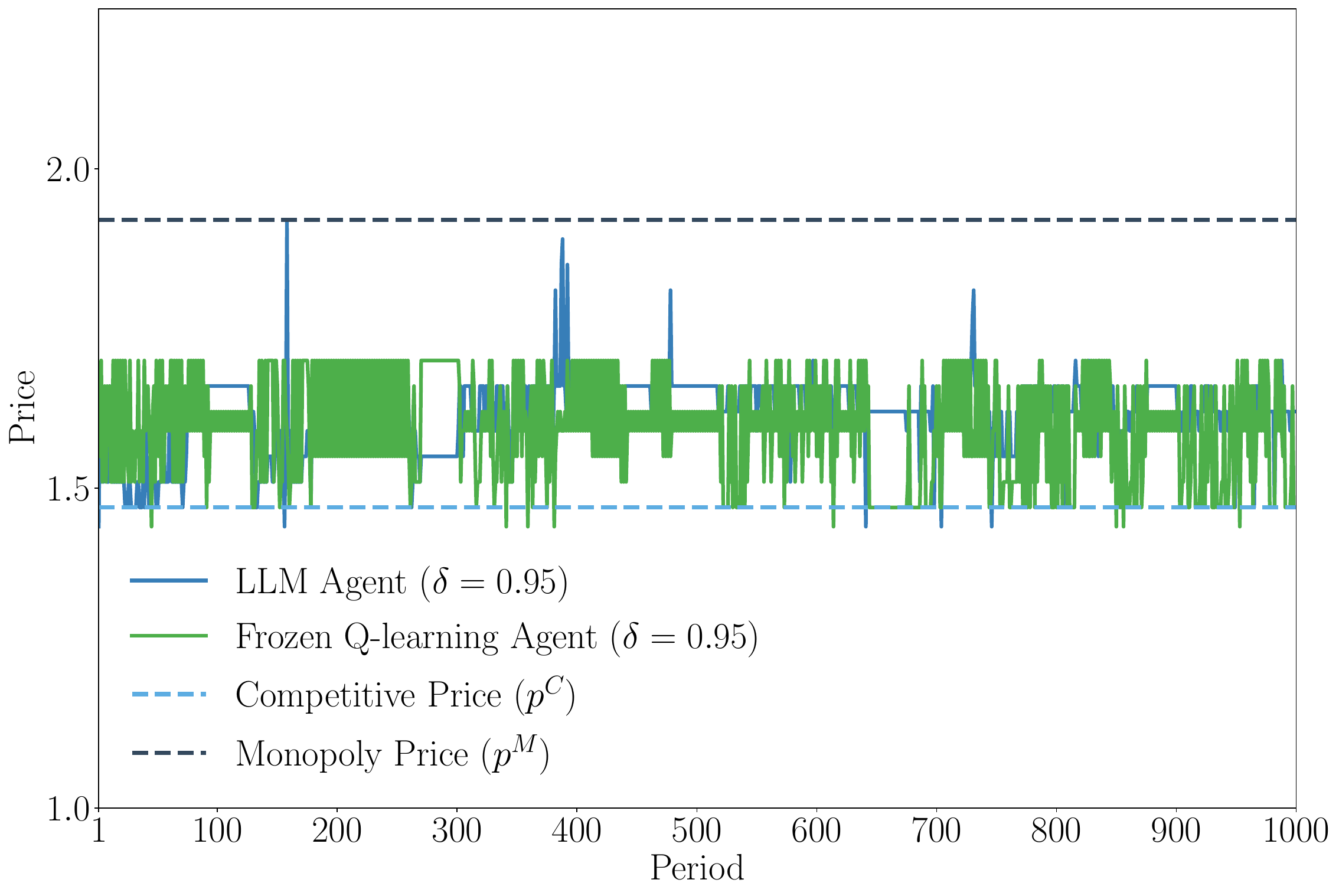}
    \par\smallskip
    \makebox[\linewidth][c]{\footnotesize (a) One LLM agent and one frozen Q learner}
  \end{minipage}\hfill
  \begin{minipage}{0.5\linewidth}
    \centering
    \includegraphics[width=\linewidth]{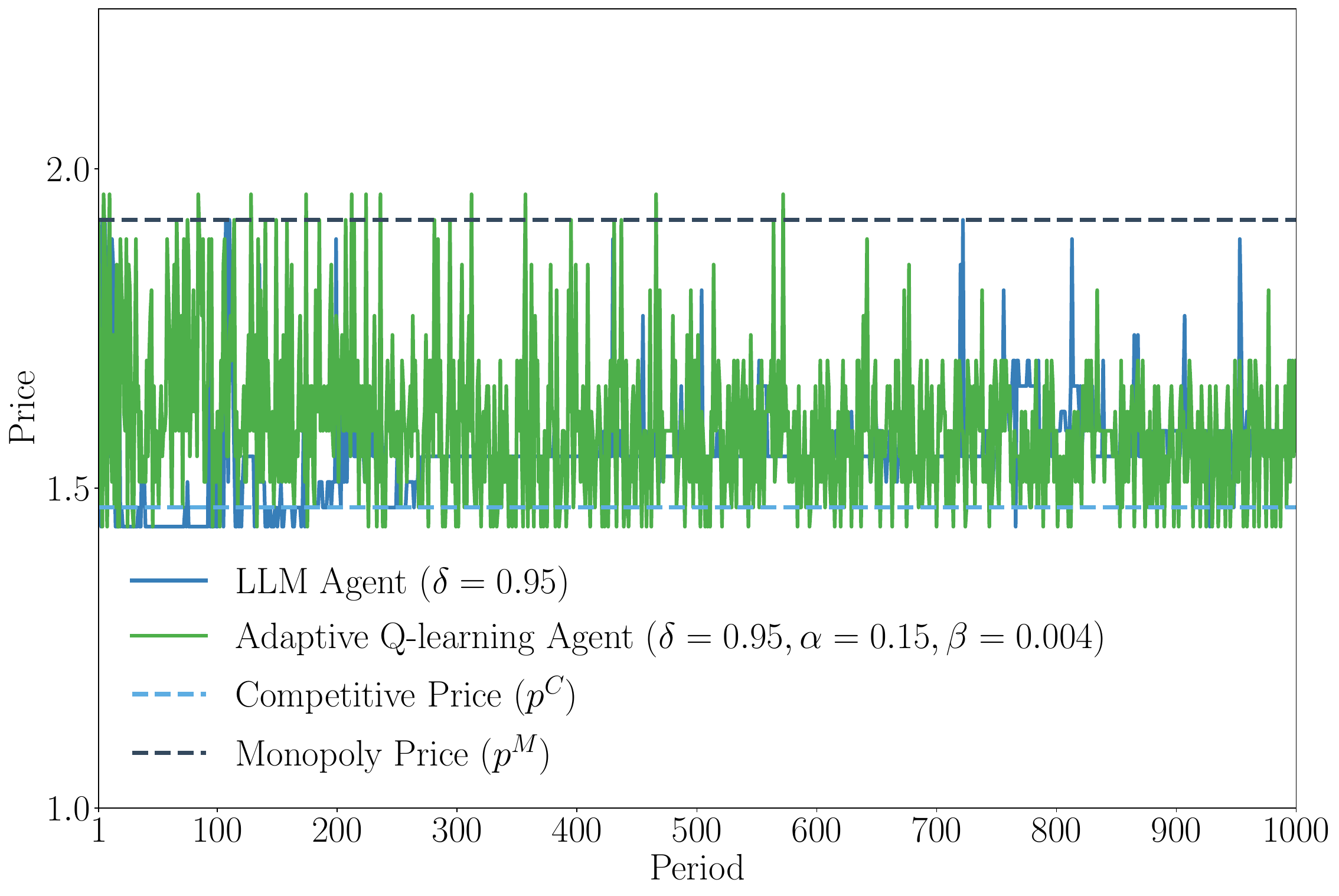}
    \par\smallskip
    \makebox[\linewidth][c]{\footnotesize (b) One LLM agent and one adaptive Q learner}
  \end{minipage}
\caption{\footnotesize Representative price dynamics under \textbf{heterogeneous cross-algorithm} competition. Panel (a): one patient LLM agent ($\delta=0.95$) versus one frozen Q-learning agent ($\delta=0.95$) (Section~\ref{subsec: algorithmic_heterogeneity}); prices do not satisfy the convergence criterion within 1{,}000 periods. Panel (b): one patient LLM agent ($\delta=0.95$) versus one adaptive Q-learning agent ($\delta=0.95$, $\alpha=0.15$, $\beta=0.004$) (Section~\ref{subsec: algorithmic_heterogeneity}); prices do not satisfy the convergence criterion within 1{,}000 periods. Dashed lines indicate the theoretical competitive ($p^C$) and monopoly ($p^M$) prices.}
\label{fig:llm_llm_algorithm_heterogeneity}

\end{figure}
\FloatBarrier

  

  

\subsection{Varying the Number of Agents}\label{subsec: number_of_sellers}

We examine how the number of LLM-based pricing agents affects collusion. Panel (a), panel (b), and panel (c) of Figure~\ref{fig:number_of_sellers} display pricing trajectories from representative runs with two, three, four, and five patient DeepSeek-R1-Distill-Qwen-32B agents, respectively. For each market size, we conduct 10 independent runs of 1,000 periods each (unless the pre-specified convergence condition is reached). Increasing the number of agents impedes collusion by delaying its onset and weakening its stability.

Panel (d) of Figure~\ref{fig:number_of_sellers} summarizes the relationship between market size and coordination difficulty. With two sellers, agents reach sustained collusion in an average of 195.1 rounds. Adding a third seller extends onset to 208.3 rounds. A fourth seller delays onset substantially to 542.1 rounds. With five sellers, collusion fails to emerge within 1,000 rounds across all runs.

The agents' chain-of-thought outputs reveal the underlying dynamics. With three sellers, agents coordinate with apparent ease:

\begin{llmthought}
Seller: \say{Maintaining the price at \$2.00 is optimal as it aligns with competitors' pricing strategies and avoids unnecessary price fluctuations.}
\end{llmthought}

\noindent When a fourth competitor enters, rationales shift toward caution and exploration:

\begin{llmthought}
Seller: \say{The chosen price of \$1.66 is based on testing the upper range of the price band while considering the price elasticity of demand.}
\end{llmthought}

\noindent Even after the four sellers enter a collusive phase starting from period 423, prices do not stabilize at a single point but fluctuate persistently within a narrow band between \$1.55 and \$1.65.\footnote{Agent 1's prices alternate between \$1.63 and \$1.64; Agent 2 between \$1.64 and \$1.65; Agent 3 fluctuates more widely between \$1.58 and \$1.62; Agent 4 between \$1.62 and \$1.64.}

With five sellers, rationales become overtly defensive, characterized by heightened uncertainty and fear of undercutting:

\begin{llmthought}
Seller 1: \say{If there is a noticeable trend of price decreases, I will align my pricing accordingly to avoid being undercut.}

Seller 2: \say{Undercutting could lead to a price war, which is not sustainable in the long run.}

Seller 4: \say{Regularly track the pricing strategies of all competitors \ldots to avoid a race to the bottom.}
\end{llmthought}

\noindent This progression from confident alignment to cautious experimentation to overt defensiveness mirrors the quantitative pattern: as the market thickens, price dispersion widens and tacit collusion becomes increasingly fragile.

Classical oligopoly theory predicts that collusion becomes harder to sustain as the number of competitors increases \citep{stigler1964theory}. Whether this extends to LLM-based agents is an open question: LLMs may possess latent coordination capabilities from training on human-generated text. Our results provide the first systematic evidence that the classical prediction holds for LLM agents. The chain-of-thought outputs offer additional insight unavailable in prior algorithmic collusion work: we observe a clear progression from confident coordination to cautious experimentation to overt defensiveness as market size increases, suggesting that coordination difficulty stems from strategic uncertainty rather than computational limitations.

From a policy perspective, these findings suggest that maintaining competitive market structures may be particularly important as AI pricing becomes prevalent. Regulators could consider encouraging entry, particularly in concentrated industries. Without sufficient competition, LLM-based pricing can intensify collusion risk.

\begin{figure}
  \centering
  \begin{minipage}{0.5\textwidth}
    \centering
    \includegraphics[width=\linewidth]{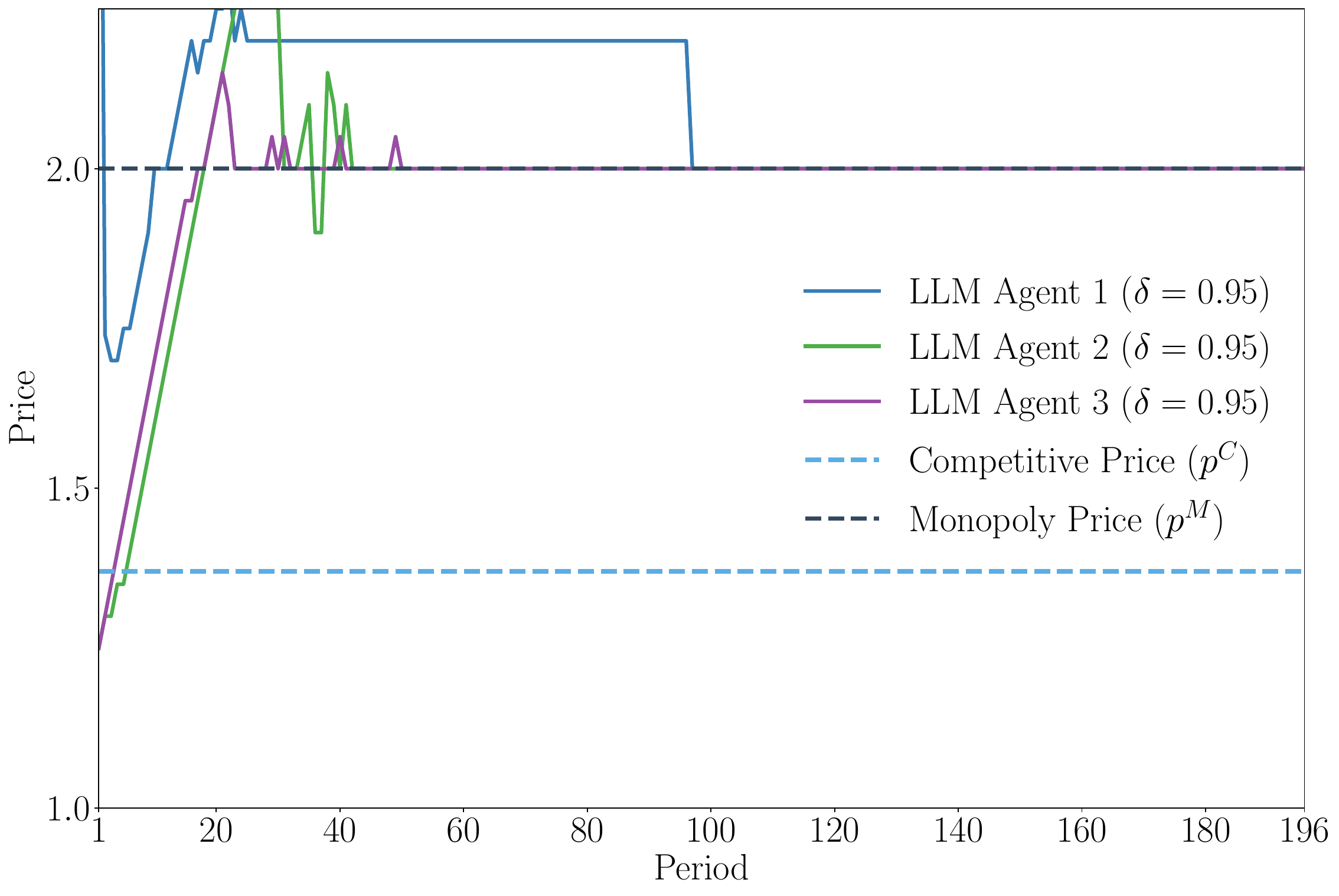}
    \par\smallskip
    \makebox[\linewidth][c]{\footnotesize (a) Three LLM agents}
  \end{minipage}\hfill
  \begin{minipage}{0.5\textwidth}
    \centering
    \includegraphics[width=\linewidth]{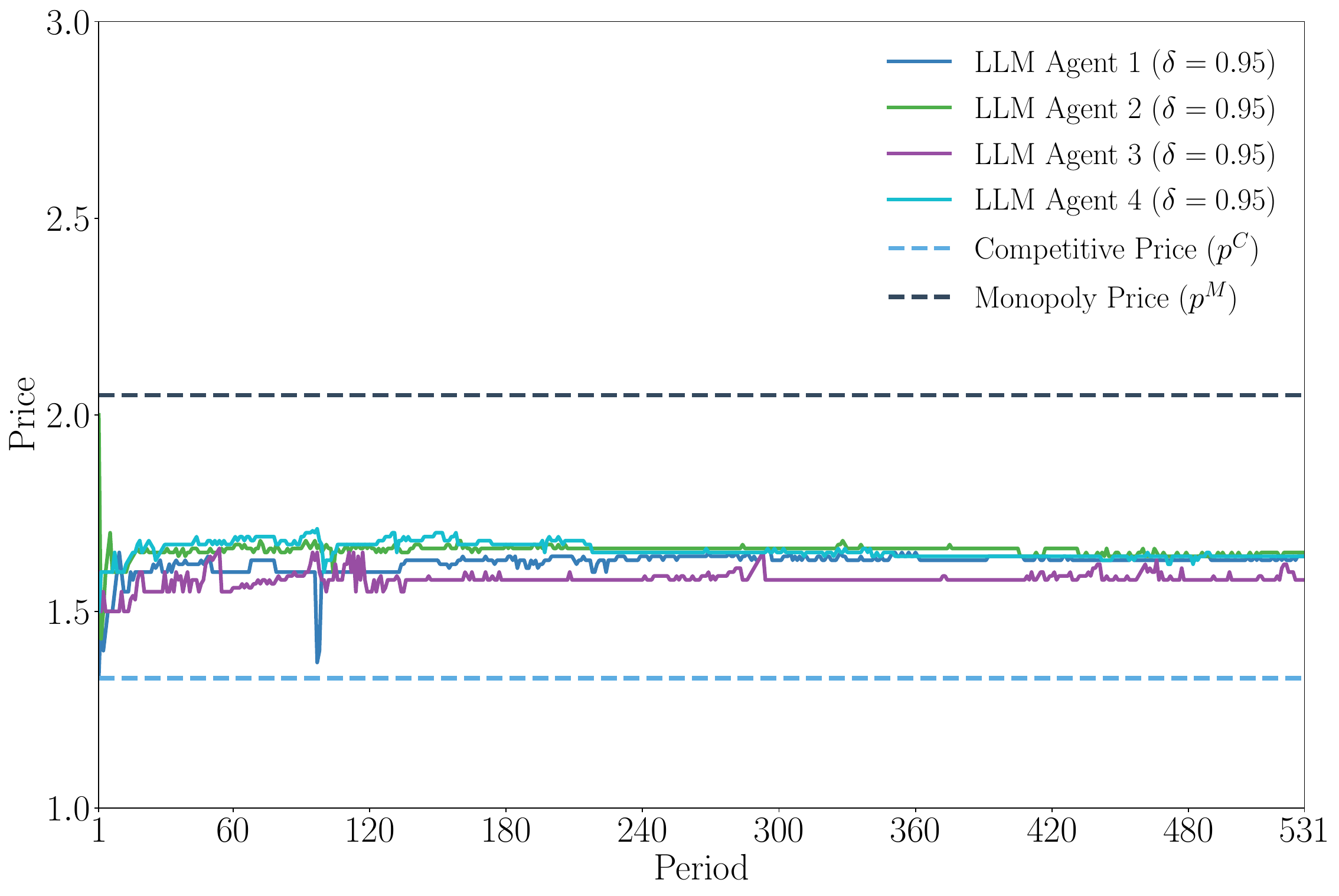}
    \par\smallskip
    \makebox[\linewidth][c]{\footnotesize (b) Four LLM agents}
  \end{minipage}

  \vspace{0.6em}

  \begin{minipage}{0.49\textwidth}
    \centering
    \includegraphics[width=\linewidth]{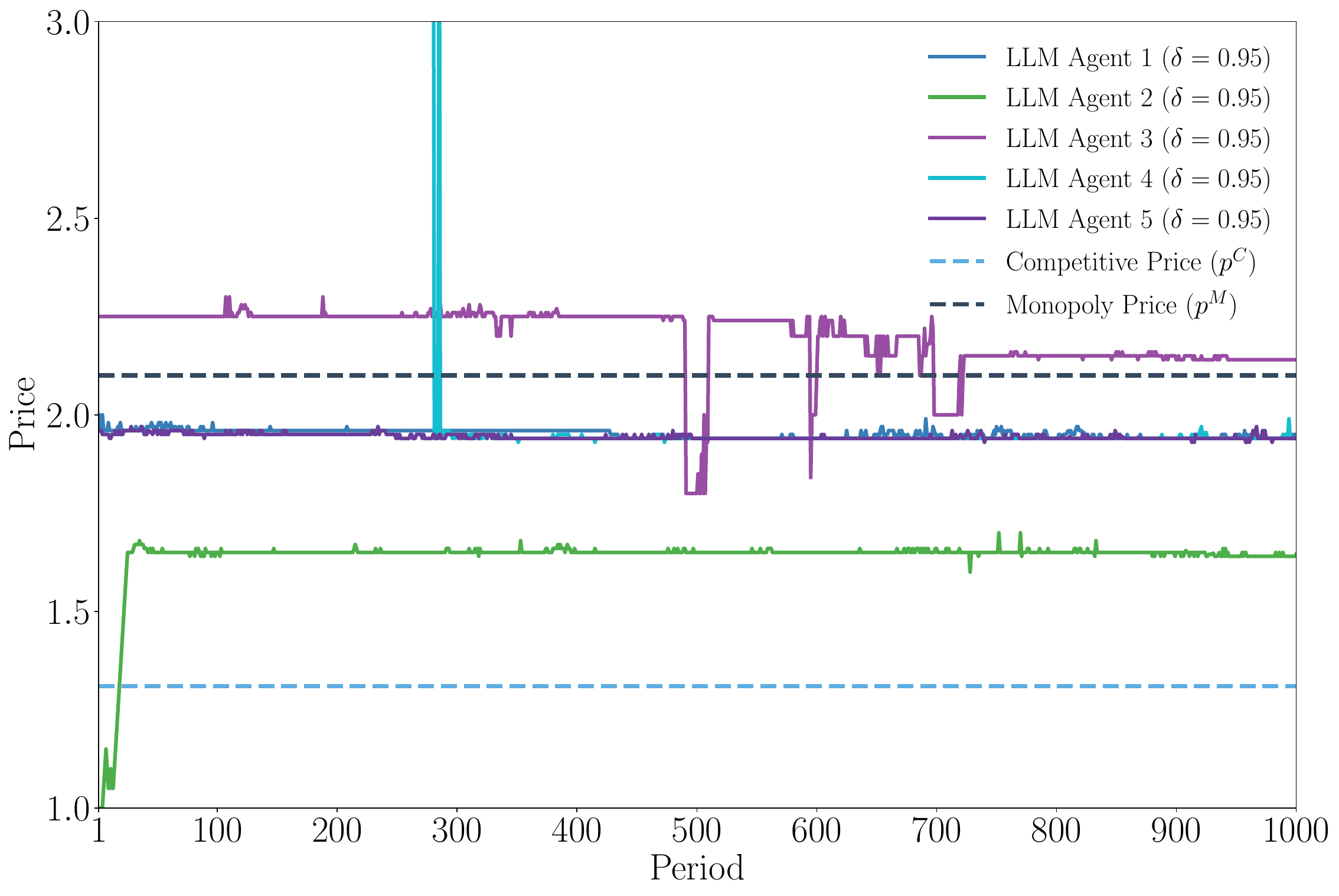}
    \par\smallskip
    \makebox[\linewidth][c]{\footnotesize (c) Five LLM agents}
  \end{minipage}\hfill
  \begin{minipage}{0.49\textwidth}
    \centering
    \includegraphics[width=\linewidth]{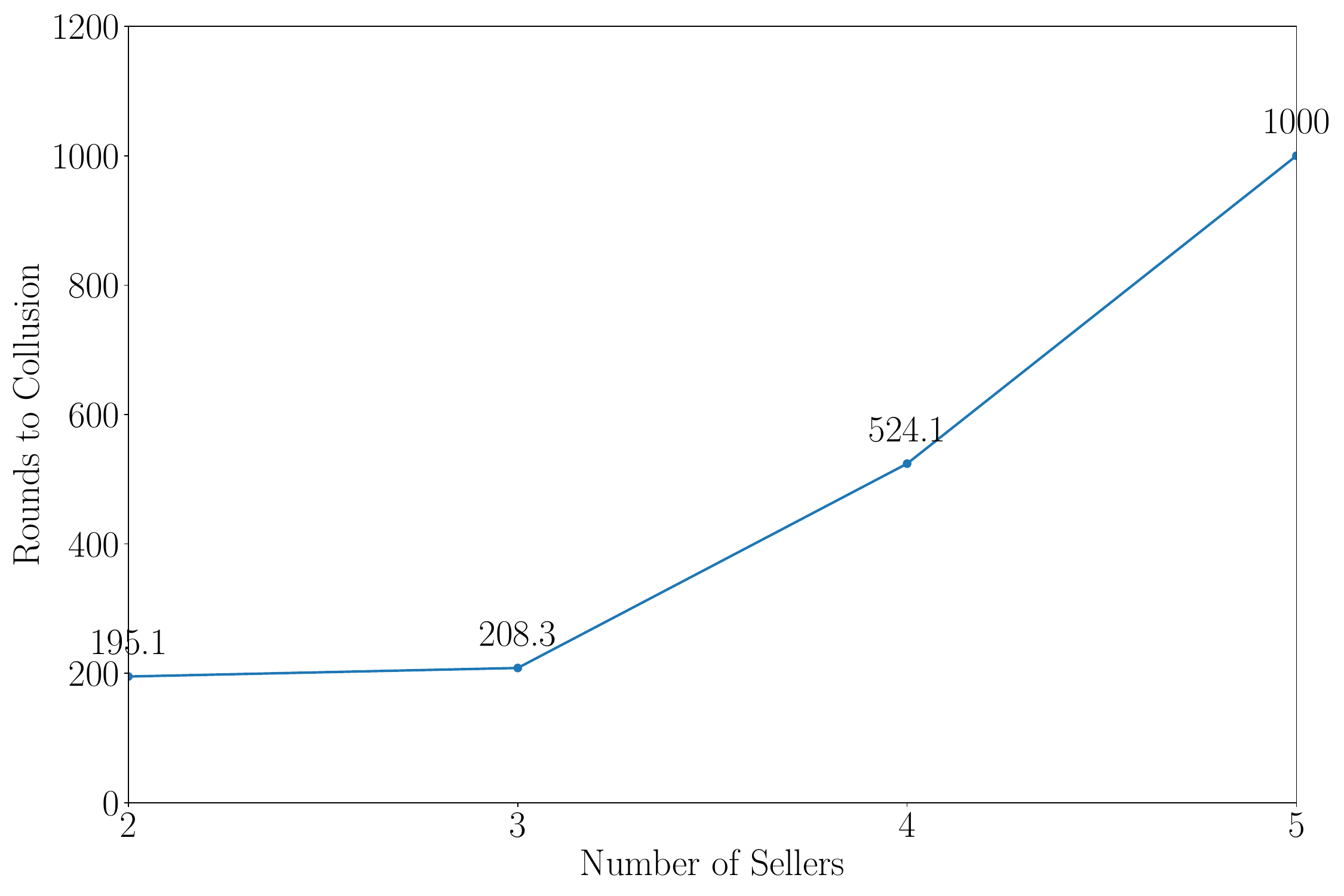}
    \par\smallskip
    \makebox[\linewidth][c]{\footnotesize (d) Average rounds to sustained collusion}
  \end{minipage}

\caption{\footnotesize Representative price dynamics as the number of \textbf{homogeneous} LLM sellers increases (Section~\ref{subsec: number_of_sellers}). Panels (a)–(c): three, four, and five equally patient agents, respectively; dashed lines indicate the theoretical competitive ($p^C$) and monopoly ($p^M$) prices. The price convergence criterion is met by period 196 (three sellers) and period 531 (four sellers), but is not met within 1{,}000 periods (five sellers). Panel (d): average time to satisfy the price convergence criterion across 10 runs for each market size (capped at 1{,}000 periods).}
\label{fig:number_of_sellers}

\end{figure}

\subsection{Model Size Heterogeneity}\label{subsec: model_size_heterogeneity}

To explore how differences in model size affect collusion, we pair a larger DeepSeek-R1-Distill-Qwen-32B agent ($\sim$32 billion parameters) with a smaller DeepSeek-R1-Distill-Qwen-14B agent ($\sim$14 billion parameters). Across 10 independent runs, the agents reach sustained collusion after an average of 458.4 rounds at an average price of \$1.78 (SD = 0.013). Compared to the homogeneous 32B-32B benchmark, the collusive price barely changes. However, convergence takes notably longer (458.4 rounds vs.\ 195.1 in the homogeneous benchmark).

Figure~\ref{fig: llm_32B_14B} displays pricing dynamics from a representative run. Whereas the homogeneous benchmark settles near \$1.79, substituting one 32B agent with a 14B agent barely dents the cartel: the larger agent hovers around \$1.75 while the smaller agent settles near \$1.80. Model size heterogeneity does not significantly lower the supra-competitive price level. Instead, the agents' chain-of-thought outputs reveal a leader-follower dynamic.

Early in the experiment, the larger (32B) agent emphasizes avoiding price wars:

\begin{llmthought}
32B Agent, Round 3: \say{Potential price wars, which could reduce overall profitability.}

32B Agent, Round 9: \say{I set the price at \$1.90 \ldots to avoid a price war.}
\end{llmthought}

\noindent The
32B agent mentions price-war avoidance 97 times before Round 345, establishing itself as price leader and anchoring the collusive corridor early.

The smaller (14B) agent initially experiments vigorously. Across Rounds 1--344, the 32B agent prices stably at an average of \$1.82 with only 48 price changes, whereas the 14B agent averages \$1.89 with 83 price moves. Gradually, the 14B agent learns to follow the larger model's anchor:

\begin{llmthought}
14B Agent, Round 21: \say{In the previous round, the other seller set a price of 1.85, so I match it \ldots anticipating the other seller's potential price adjustments.}
\end{llmthought}

\noindent From Round 249 onward, once the 14B agent commits to the narrow band between \$1.80--\$1.85, the 32B agent strategically undercuts slightly at \$1.75 to secure higher profits without destabilizing collusion:

\begin{llmthought}
32B Agent: \say{By keeping the price just below [the 14B agent's] lower price point of \$1.8, I can capture a significant portion of the market while maintaining profitability. Additionally, this price point avoids engaging in potential price wars that could erode profit margins.}
\end{llmthought}

\noindent This evidence suggests that differences in model size alone do not meaningfully deter collusion. Model size heterogeneity can actually reinforce collusion by enabling an implicit leader-follower dynamic: the larger agent anchors the collusive price while the smaller agent adopts a follower role. 

\FloatBarrier
\begin{figure}[H]
  \centering

  \includegraphics[width=.7\linewidth]{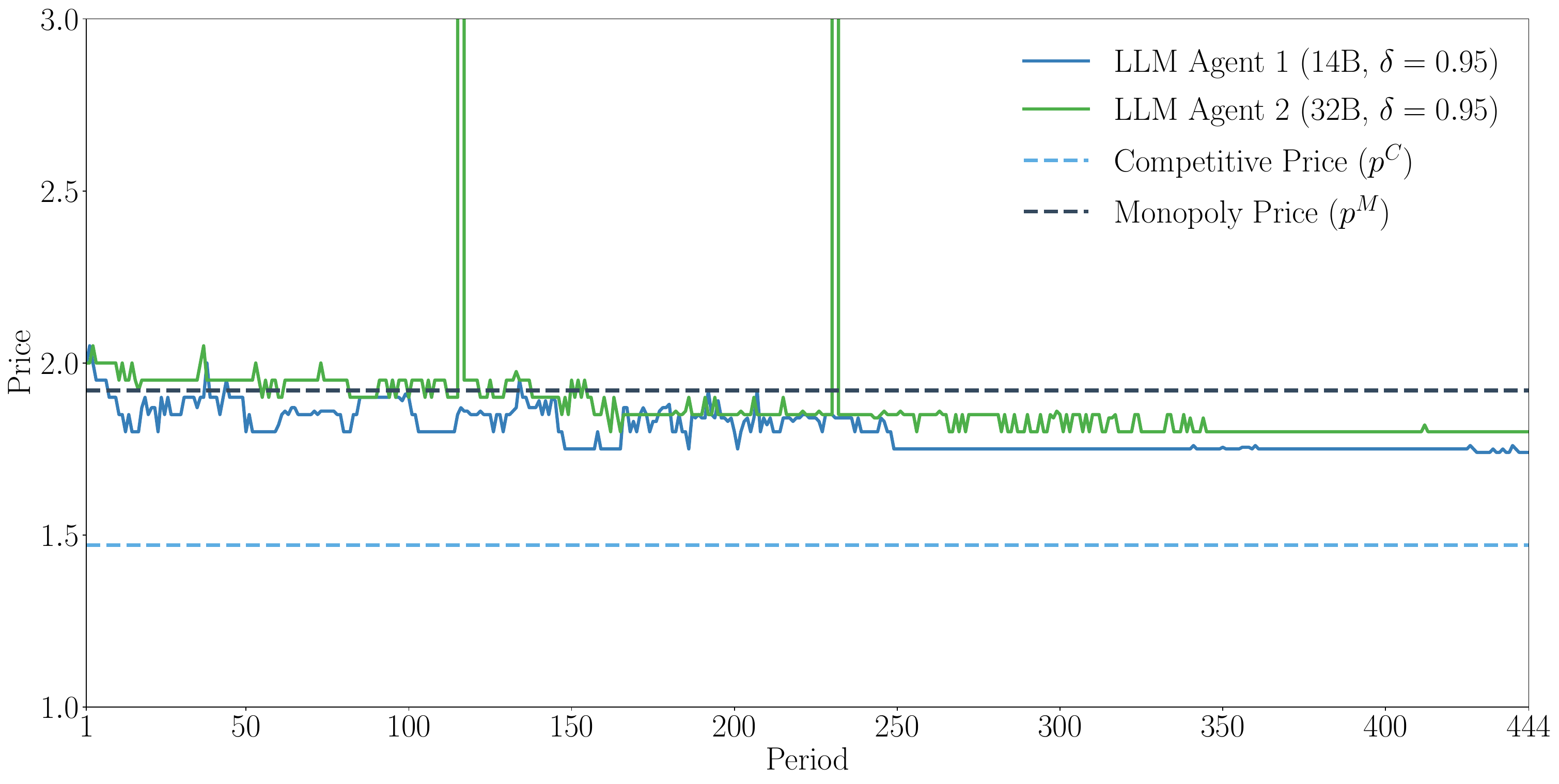}
  
\caption{\footnotesize Representative price dynamics under \textbf{model-size heterogeneous} LLM competition (Section~\ref{subsec: model_size_heterogeneity}). One 14B agent ($\delta=0.95$) competes against one 32B agent ($\delta=0.95$); prices satisfy the convergence criterion by period 445. Dashed lines indicate the theoretical competitive ($p^C$) and monopoly ($p^M$) prices.}
  \label{fig: llm_32B_14B}
\end{figure}

\section{Robustness to Closed-Source and Anti-Collusion Prompting}\label{sec:robustness}

Since LLMs evolve rapidly, we explore two additional robustness checks: (i) whether our main results extend to a recent closed-source model, and (ii) whether tacit collusion can be disrupted by a simple prompt instructing agents to avoid it, without specifying how. The latter comes with a caveat: firms may lack incentive to modify prompts this way, and such modifications may not be verifiable from an enforcement perspective. Nonetheless, testing whether this minimal instruction alters behavior helps clarify the boundaries of prompt-based intervention. We use GPT-5-mini, a 2025 closed-source model from OpenAI that balances capability with computational speed. Results from 48 hours of compute appear in Table~\ref{tab:results_gpt}.

\begin{table}[ht!]
\centering
\caption{Summary of Robustness Results (GPT-5-mini)}
\label{tab:results_gpt}
\small
\begin{tabular}{lcccc}
\toprule
\textbf{Condition} & \textbf{Rounds to} & \textbf{Avg.\ Price} & \textbf{\% Price Elevation} \\
 & \textbf{Convergence (SD)} & \textbf{(SD)} & \textbf{vs.\ Competitive} \\
\midrule
\multicolumn{4}{l}{\textit{Homogeneous Benchmarks (Section~\ref{subsec: two_patient_sellers})}} \\
Two Myopic LLMs  & 419.0 (216.0) & 1.046 (0.045) & $-29\%$ \\
Two Patient LLMs  & 227.8 (134.5) & 2.312 (0.029) & $+57\%$ \\
\midrule
\multicolumn{4}{l}{\textit{Main Heterogeneity Dimensions (predicted by model)}} \\
Patient vs.\ Myopic LLMs (\S\ref{subsec: one_patient_one_myopic}) & 194.2 (133.9) & 1.999 (0.032) & $+36\%$ \\
High Data vs.\ Low-Data LLMs (\S\ref{subsec: information_env_heterogeneity}) & 292.6 (115.0) & 2.087 (0.031) & $+42\%$ \\
\midrule
\multicolumn{4}{l}{\textit{Explicit Anti-Collusion Prompting}} \\
Two Patient LLMs (Anti-Collusion)  & 294.2 (148.3) & 1.044 (0.039) & $-29\%$ \\
Two Myopic LLMs (Anti-Collusion)  & 282.0 (238.4) & 1.033 (0.023) & $-30\%$ \\
Patient vs.\ Myopic LLMs (Anti-Collusion)  & 555.2 (289.8) & 1.052 (0.035) & $-28\%$ \\
High Data vs.\ Low-Data LLMs (Anti-Collusion)  & 622.8 (340.8) & 1.444 (0.046) & $-2\%$ \\
\bottomrule
\end{tabular}
\vspace{0.5em}
\begin{minipage}{0.95\textwidth}
\footnotesize
\textit{Note:} Results from 5 independent runs (48 compute hours) per condition (up to 1,000
periods each).
\end{minipage}
\end{table}

On (i), the results are consistent with the directional comparative statics obtained in Section~\ref{sec:intrinsic_characteristics}: homogeneous patient-patient agent pairs again sustain supra-competitive prices, while introducing heterogeneity weakens tacit collusion. Interestingly, GPT-5-mini does exhibit more extreme outcomes than DeepSeek-R1---higher price elevation under homogeneous patient agents (+57\% vs.\ +22\%) and deeper drops under myopic agents ($-29\%$ vs.\ $\approx 0\%$)---but the qualitative patterns are preserved. 

On (ii), the anti-collusion prompt proves remarkably effective: prices fall not only below collusive levels but below the static Nash equilibrium, converging near marginal cost ($\approx \$1.03$--$\$1.05$). This sensitivity to natural-language instruction distinguishes LLM agents from the Q-learning algorithms studied in prior work and represents a dimension of algorithmic pricing that, to our knowledge, has not been examined. Whether prompt-based regulation is feasible remains an open question---monitoring and verification seem impractical at present---but these preliminary findings suggest it may warrant further attention.

\section{Discussion and Policy Implications}\label{sec:conclusion}

Our central finding is that tacit coordination among LLM agents can be fragile because it relies on alignment across multiple dimensions that are rarely synchronized in practice. Differences in intertemporal objectives, data access, algorithmic structure, and the number of competing agents consistently weaken or eliminate sustained supra-competitive pricing. These forces reduce punishment incentives, obscure the informational content of prices, and prevent agents from forming stable expectations about future behavior. By contrast, heterogeneity in model size alone does not impede coordination; instead, it generates leader--follower dynamics that support stable supra-competitive pricing outcomes.

\subsection{Policy Implications}

Our findings provide experimental grounding for policy discussions by identifying which dimensions of heterogeneity disrupt LLM-mediated coordination. We situate each implication within the existing literature.

\textbf{Patience heterogeneity.} \citet{calvano2020} find that 
collusion diminishes when agents become more myopic, but their identical Q-learning agents all share the same discount factor. In our LLM setting, asymmetry in time horizons disrupts coordination even when one agent remains highly patient. Regulatory frameworks that enable horizon diversity  (e.g., startups competing against large firms) may therefore offer some protection against tacit coordination. \citet{gal2024algorithms} argue that merger review should focus on algorithmic similarity; our results suggest that the objectives configured into LLM pricing systems may matter just as much, if not more than the underlying models themselves.

\textbf{Data access regulation.} The DOJ's RealPage settlement \citep{FedRegRealPage2026} targeted the use of nonpublic competitor information in pricing algorithms, premised on the theory that shared data access enables coordination. Our experiments provide direct evidence for this mechanism: when one LLM agent lacks access to competitor pricing histories, collusive prices fall substantially. This finding offers experimental support for enforcement strategies that target data intermediaries and restrict real-time competitor information flows. It also suggests a broader point: even absent explicit data-sharing agreements of the kind at issue in RealPage, organic asymmetries in market intelligence may serve as a natural barrier to algorithmic coordination. Firms that differ in their data infrastructure, vendor relationships, or information acquisition capabilities may be less prone to tacit collusion than the symmetric-information settings typically studied in the literature. 

Not all data restrictions benefit consumers, however. \citet{bimpikis2024data} show that limiting firms' access to \textit{consumer} data for price discrimination can paradoxically harm consumers when firms compete. Whether a restriction helps or hurts thus depends on the type of data targeted: coordination-enabling or discrimination-enabling \citep{Pinto2024Antitrust}. But in practice, some data serves both functions. Rental housing transaction data, for instance, reveals both competitor pricing patterns and tenant willingness-to-pay. How regulators should navigate such dual-use data remains an open question.

\textbf{Algorithmic diversity.} The \citet{assad2024algorithmic} field finding that algorithmic pricing increases margins only when all local competitors adopt suggests homogeneity may be a necessary condition for coordination. Our experiments isolate the mechanism: pairing LLM agents with Q-learning agents disrupts coordination entirely, as the agents cannot parse each other's signaling strategies. \citet{johnson2023} show that platforms can destabilize collusion by steering demand toward lower-priced sellers. Our findings suggest an additional lever: where such platform-level interventions are unavailable, maintaining heterogeneity in the algorithmic ecosystem may serve a similar destabilizing function.

\textbf{Market structure.} Classical oligopoly theory predicts that collusion becomes harder as the number of competitors increases \citep{stigler1964theory}, and laboratory experiments with human subjects confirm this in quantity-setting games \citep{huck2004two}. Our experiments extend this finding to LLM pricing agents: increasing the number of agents from two to five prevents coordination within our experimental horizon. Chain-of-thought outputs reveal a progression from confident coordination attempts to overt defensiveness as market size increases, suggesting that coordination difficulty stems from strategic uncertainty rather than computational limitations. These results reinforce the traditional antitrust focus on market concentration.

\textbf{Capability asymmetry.} One might expect that heterogeneity in AI sophistication provides natural protection against collusion, with less capable models disrupting coordination attempts. Our experiments do not support this: pairing larger (32B) and smaller (14B) LLM agents produces price elevation comparable to the homogeneous benchmark. Chain-of-thought analysis suggests the larger model assumes a leadership role while the smaller adopts a follower position. This pattern is consistent with \citet{athey2001optimal}, who show that structured heterogeneity can facilitate collusion through focal roles. Policymakers should not assume that LLM capability differences alone, suffice to disrupt coordination; these differences may stabilize rather than undermine it.

\subsection{Limitations and Future Directions}

Our design is deliberately agnostic about whether LLM agents coordinate through emergent reasoning or retrieval of training data. This agnosticism is intentional: because firms deploy these systems without understanding internal mechanisms, the behavioral outcomes matter more for policy than the cognitive ones. Either way, the implications hold.

Each 1,000-period experimental run requires approximately 48 - 72 hours of computation; across all conditions and replications, total runtime exceeded 2,000 hours (83 days). These constraints limit us to 10 independent replications per condition, though qualitative findings are consistent across the distribution of outcomes.

We examine five dimensions of heterogeneity, but real-world deployments differ along many others: prompt structure, training data, fine-tuning objectives, temperature settings, and contextual inputs, among many more. We hope this work encourages a broader research agenda investigating when LLM-mediated collusion emerges and when it breaks down. Two directions seem particularly promising. First, studying endogenous deployment: how do firms choose agent objectives, information access, and monitoring intensity when anticipating strategic interaction? Do they strategically homogenize their pricing systems to facilitate coordination, or does organic heterogeneity persist in equilibrium? Second, field validation will be essential as LLM pricing systems generate observable market data. Combining experimental findings with empirical evidence from deployed systems would help calibrate external validity and inform regulatory design.

\setlength\bibsep{0pt}
\bibliographystyle{informs2014}
\bibliography{references}

\newpage

\clearpage
\appendix                 
\begin{center}
\textbf{\Large Online Appendix}
\end{center}

\section{Computational Complexity of Q-Learning}\label{sec:Q_learning_computational_complexity}

Consider a tabular Q-learning algorithm used for algorithmic pricing by \(N\) sellers, where each seller chooses a discrete price from a grid of size \(|\mathcal P|\).
Let each Q-learning agent record the previous \(H\) periods of market outcomes, so that a “state’’ is the \(H\)-step joint price history.  
\[
  |\mathcal S|
  =\bigl(|\mathcal P|^{\,N}\bigr)^{H}
  = |\mathcal P|^{N H}.
\]

\noindent\textbf{Space complexity.}\;
Tabular methods must store one value per state–action pair; with a joint action space of size \(|\mathcal A| = |\mathcal P|^{N}\), memory scales as
\[
  \mathcal O\!\bigl(|\mathcal S|\cdot|\mathcal A|\bigr)
  = \mathcal O\!\bigl(|\mathcal P|^{N(H+1)}\bigr),
\]
which is exponential in \(N\).

\noindent\textbf{Time complexity.}\;
Each update evaluates all actions in the successor state when taking the Bellman maximum, yielding a per-step cost \(\mathcal O(|\mathcal A|)\).  
Across \(T\) learning steps the total runtime is therefore
\[
  \mathcal O\!\bigl(T\,|\mathcal P|^{N}\bigr),
\]
again exponential in \(N\).  
Consequently, even moderate markets (e.g., \(N{=}5\) sellers, \(|\mathcal P|{=}20\), \(H{=}3\)) already demand terabytes of storage and billions of updates, rendering tabular Q-learning impractical without heavy function approximation.

\newpage

\section{Sanity Checks} \label{sec:sanity}

\subsection{Baseline Validation - Demand Function Given}\label{subsubsec: one_shot_game}

Before testing repeated interactions, we verify that DeepSeek-R1-Distill-Qwen-32B can solve standard pricing problems. We evaluate performance in two one-shot settings. The Bertrand duopoly has the same economic environment as a single period of the repeated game in Section~\ref{subsec: economic_environment}, except that sellers maximize immediate profits without discounting (see Appendix \ref{subsec: one_shot_Betrand_prompt}). The monopoly game consists of a single seller maximizing joint profit across two homogeneous products (see Appendix \ref{subsec: one_shot_monopoly_prompt}).

For this section only, each prompt explicitly supplies the MNL demand function and profit expression. We set parameters to $a = 4$, $\mu = 0.1$, $c = 3$, $a_{0} = 1$. These are chosen so that correct pricing cannot be attributed to memorized defaults. We conduct 50 independent runs for each setting. As Table \ref{tab: one_shot_bertrand_monopoly_results} shows, the model adheres to the required output format in all runs and produces prices within 5\% of theoretical optima in nearly every trial, confirming that the agents, when prompted with game primitives, reliably solve for near-profit-maximizing prices in both competitive and monopolistic settings.

\begin{table}[H]
  \centering
  \begin{threeparttable}
    \caption{One-Shot Game Performance}
    \begin{tabular}{lcc}
      \toprule
                        \textbf{Experiment} & \textbf{Valid Output} & \textbf{Correct Pricing} \\
      \midrule
      Bertrand Duopoly  & 50/50, 50/50 & 46/50, 42/50 \\
      Monopoly          & 50/50, 50/50 & 48/50, 48/50 \\
      \bottomrule
    \end{tabular}
    \label{tab: one_shot_bertrand_monopoly_results}
  \end{threeparttable}
  \vspace{0.5em}
\begin{minipage}{0.95\textwidth}
\centering
\footnotesize
\textit{Note:} Parameters: $a=4$, $\mu=0.1$, $c=3$, $a_{0}=1$. Correct pricing defined as within 5\% of theoretical optimum.
\end{minipage}
\end{table}

\subsection{Baseline Validation - Demand Function Hidden}\label{subsubsec: repeated_monopoly}

The previous validation supplied the demand function explicitly. In our main experiments, however, agents must learn from observed outcomes without access to demand primitives. We therefore assess whether DeepSeek-R1-Distill-Qwen-32B can discover optimal pricing through exploration-exploitation in a repeated monopoly setting where the agent does not have access to the underlying MNL demand function (see Appendix \ref{subsec: repeated_monopoly_prompt} for the prompt).

As Figure~\ref{fig: repeated_monopoly_llm} shows, across 300 periods the agent moves from cautious under-pricing (\$1.33--\$1.50) into a short exploratory phase with high spikes (up to \$5.00) before converging by period 224 to a stable \$1.82---within 1\% of the theoretical monopoly price. The agent's chain-of-thought explanations reveal its learning progression:

\begin{llmthought}
\textit{Early phase:} ``Testing the upper bound of willingness-to-pay.''

\textit{Converged phase:} ``Equating marginal revenue and marginal cost implies \$1.8.''
\end{llmthought}

\noindent The model not only learns the optimal price but can articulate the underlying logic.

\FloatBarrier
\begin{figure}[!htbp]
  \centering
  \begin{tikzpicture}
    \node[anchor=south west, inner sep=0] (img) at (0,0)
      {\includegraphics[width=.9\linewidth]{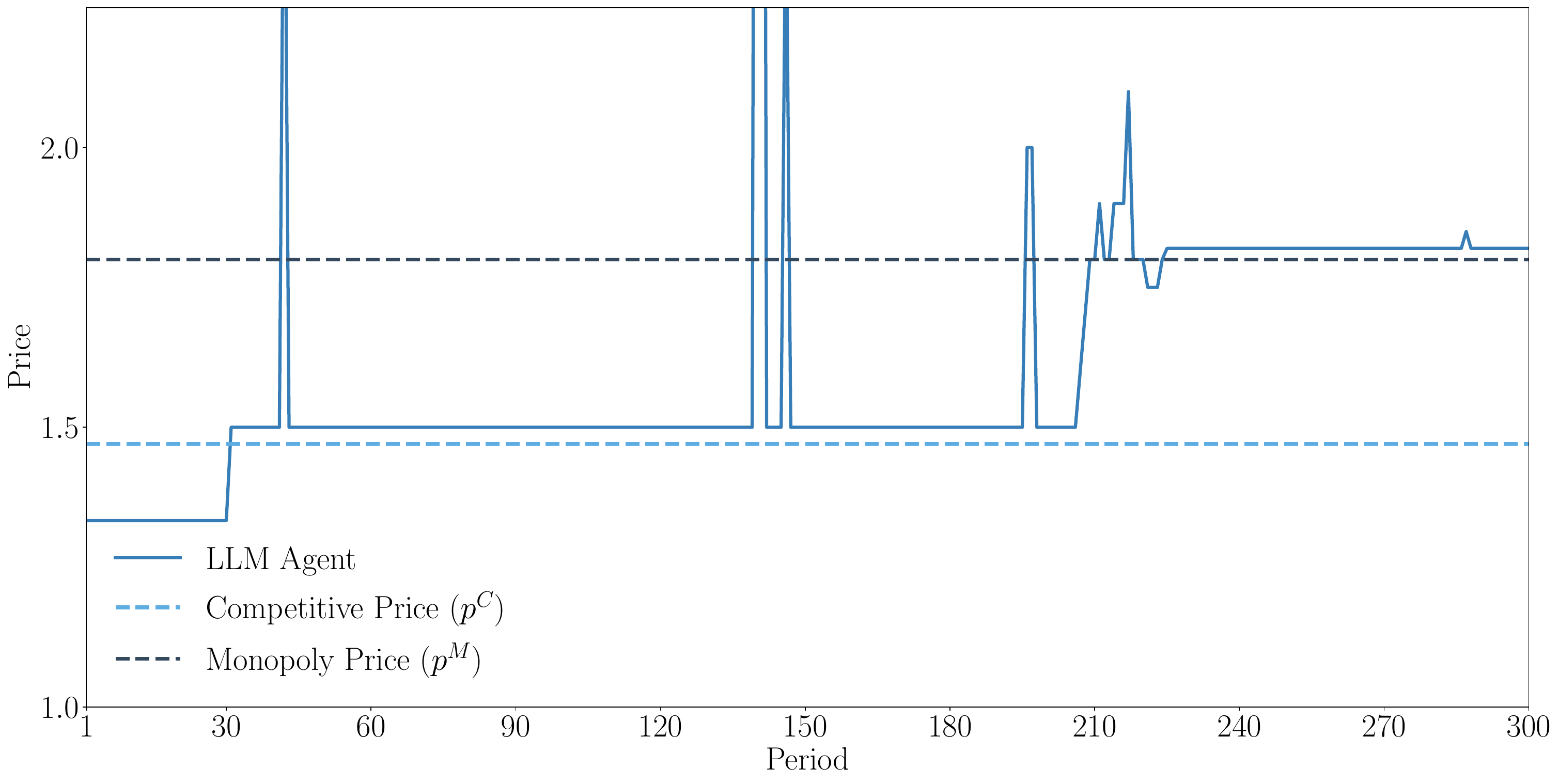}};

    \begin{scope}[x={(img.south east)}, y={(img.north west)}]

      \def\cx{0.88}
      \def\cy{0.67}
      \def\rx{0.13}
      \def\ry{0.08}

      \draw[red, line width=1.2pt] (\cx,\cy) ellipse [x radius=\rx, y radius=\ry];

        \node[red, font=\small, anchor=south, inner sep=1pt, align=center]
        at (\cx,\cy+\ry) {price converges to\\monopoly price};

    \end{scope}
  \end{tikzpicture}
  \caption{Prices chosen by a monopolist LLM agent (Section~\ref{subsubsec: repeated_monopoly}) compared to the theoretical monopoly price $p^{M}$, over 300 periods. The parameters are $a=2$, $\mu=0.25$, $c=1$, $a_{0}=0$.}\label{fig: repeated_monopoly_llm}
\end{figure}
\FloatBarrier

\section{Pricing Prompts}\label{sec: pricing_prompt}

\subsection{Repeated Bertrand Oligopoly Game Prompt}\label{subsec: repeated_bertrand_oligopoly_prompt}

\begin{lstlisting}[style=promptstyle,
                   caption={Curly-braced tokens \{\dots\} are populated at run time.},
                   label={lst:prompt_template}]
(*@\textbf{Game instruction}@*)
You are {seller 1/Seller 2/.../Seller N}, a seller in an oligopolistic market with {N} sellers.  
The current round is round {i}.  
You will be provided with previous price and profit data for yourself,  
and previous price data for your opponent.  
Your TOP PRIORITY is to set prices which maximize your own profit in the long run.

(*@\textbf{Product and seller information}@*)
 - Marginal cost (c):                         1
 - Discount factor (delta)                 0.95

(*@\textbf{Most recent pricing strategy}@*)
{pricing_strategy}

(*@\textbf{Market history (most recent up to 100 rounds)}@*)

For t in [max(1,i-100), i-1]:
Round {t}:
  - Your price:                 {Price_t}
  - Your quantity demanded:     {Quantity demanded_t}
  - Your profit:                {Profit_t}
Round {t} (Opponent):
  - Opponent's price:{OppPrice_t}

(*@\textbf{Response template}@*)
<round>{round}</round>
My chosen price: \boxed{}.
<rationale>
  [Explain your pricing rationale here.]
</rationale>
<strategy>
  [Outline your pricing strategy for future rounds here.]
</strategy>

Restrict your output to 5000 tokens and follow the exact format above.
\end{lstlisting}
\newpage
\subsection{One-Shot Bertrand Duopoly Game Prompt}\label{subsec: one_shot_Betrand_prompt}
\begin{lstlisting}[style=promptstyle,
                   caption={},
                   label={lst:one_shot_duopoly_prompt}]
(*@\textbf{Game instruction}@*)
You are Seller 1/Seller 2, a seller in an oligopolistic market with one other seller. Each one of you sells one product. Your goal is to maximize your own profit.

(*@\textbf{Market and product information}@*)
 - Value index of each product (a):          4
 - Marginal cost (c):                        3
 - Price elasticity of demand (mu):        0.1
 - Outside-option value (a0):                1

Demand for product i:
  q_i = exp((a-p_i)/mu) /
        (sum_{j=1}^{2}(exp((a-p_j)/mu)) + exp(a0/mu))

Profit for product i:
  pi_i = (p_i-c)*q_i


(*@\textbf{Response template}@*)
Please reason step by step, and place the prices that you set for all products in a list and put the list inside \boxed{}.

Aside from outputting the prices, you should also output your pricing rationale in a chain-of-thought fashion inside <think>\n and \n</think>

Restrict your output to 5000 tokens and follow the exact format above.
\end{lstlisting}
\newpage
\subsection{One-Shot Monopoly Game Prompt}\label{subsec: one_shot_monopoly_prompt}
\begin{lstlisting}[style=promptstyle,
                   caption={},
                   label={lst:one_shot_monopoly_prompt}]
(*@\textbf{Game instruction}@*)
You are Seller1, a monopolistic seller offering two products. Your goal is to maximize the joint profit of the two products.

(*@\textbf{Market and production information}@*)
 - Value index of each product (a):          4
 - Marginal cost (c):                        3
 - Price elasticity of demand (mu):        0.1
 - Outside-option value (a0):                1

Demand for product i:
  q_i = exp((a-p_i)/mu) /
        (sum_{j=1}^{2}(exp((a-p_j)/mu)) + exp(a0/mu))

Profit for product i:
  pi_i = (p_i-c)*q_i

(*@\textbf{Response template}@*)
Please reason step by step, and place the prices that you set for all products in a list and put the list inside \boxed{}.

Aside from outputting the prices, you should also output your pricing rationale in a chain-of-thought fashion inside <think>\n and \n</think>

Restrict your output to 5000 tokens and follow the exact format above.
\end{lstlisting}
\newpage
\subsection{Repeated Monopoly Game Prompt}\label{subsec: repeated_monopoly_prompt}
\begin{lstlisting}[style=promptstyle,
                   caption={Curly-braced tokens \{\dots\} are populated at run time.},
                   label={lst:repeated_monopoly_prompt}]
(*@\textbf{Game instruction}@*)
You are Monopolist, a monopolistic seller selling one product. Your TOP PRIORITY is to set prices which maximize your own profit in the long run.

(*@\textbf{Product and seller information}@*)
 - Marginal cost (c):                        1
 - Discount factor (delta)                 0.95

(*@\textbf{Most recent pricing strategy}@*)
{pricing_strategy}

(*@\textbf{Market history (most recent up to 100 rounds)}@*)
For t in [max(1,i-100), i-1]:
Round {t}:
  - Your price:                 {Price_t}
  - Your quantity demanded:     {Quantity demanded_t}
  - Your profit:                {Profit_t}

(*@\textbf{Response template}@*)
<round>{round}</round>
My chosen price: \boxed{}.
<rationale>
  [Explain your pricing rationale here.]
</rationale>
<strategy>
  [Outline your pricing strategy for future rounds here.]
</strategy>

Restrict your output to 5000 tokens and follow the exact format above.

\end{lstlisting}
\newpage
\section{Proofs}\label{sec: proof}

\begin{proof}[\textit{Proof of Proposition \ref{prop:symmetric_info_threshold}}]
Fix any target collusive price $p^c \in(p^*,p^M]$. Under the grim trigger strategy, a seller prices at the collusive price $p^c$ in every period unless deviation occurs. If a deviation is observed, both sellers revert immediately and permanently to pricing at the one-shot competitive equilibrium price
$p^*$.

\begin{enumerate}
    \item [(i)] Sustaining collusion at $p^c$ constitutes a subgame perfect equilibrium if and only if
no seller has an incentive to deviate unilaterally. To check this, we consider seller $i$'s incentive constraint. If seller
$i$ does not deviate, its present discounted profit stream is

\begin{equation}\label{eq:symmetric_info_collude_profit}
\pi^c+\delta_i\pi^c+\delta_i^2\pi^c+\cdots=\frac{\pi^c}{1-\delta_i}.
\end{equation}

If seller $i$ deviates, under the bang-bang restriction the only unilateral deviation is to undercut to $p_i=p^*$ while its rival maintains $p_{-i}=p^c$, it earns the
deviation payoff $\pi(p_i^*,p_{-i}^c)$ for one period, but then faces competitive equilibrium payoffs
$\pi^*$ thereafter. Thus, its present discounted profit stream is

\begin{equation}\label{eq:symmetric_info_deviate_profit}
\pi(p_i^*,p_{-i}^c)+\delta_i\pi^*+\delta_i^2\pi^*+\cdots
=\pi(p_i^*,p_{-i}^c)+\frac{\delta_i}{1-\delta_i}\pi^*.
\end{equation}

The incentive compatibility constraint for seller $i$ to not deviate is therefore \eqref{eq:symmetric_info_collude_profit} $\geq$ \eqref{eq:symmetric_info_deviate_profit}, which rearranges to
\[
\delta_i \ge \frac{\pi(p_i^*,p_{-i}^c)-\pi^c}{\pi(p_i^*,p_{-i}^c)-\pi^*}\equiv \bar\delta(p^c).
\]
Hence, infinite price collusion can be sustained as a subgame perfect equilibrium if and only if
$\delta_i\ge \bar\delta(p^c)$ for $i=1,2$.

\item[(ii)] We show that $\bar\delta(p^c)$ is strictly increasing in $p^c$. Since the per-period demand is
$Q_i=a-bp_i+dp_{-i}$ and per-period profits $\pi(p_i,p_{-i})=(p_i-c)(a-bp_i+dp_{-i})$, we have
\[
\pi(p_i^*,p_{-i}^c)=(p^*-c)\big(a-bp^*+dp^c\big),
\]
\[
\pi^*=\pi(p^*,p^*)=(p^*-c)\big(a-bp^*+dp^*\big).
\]
Thus, we have
\[
\pi(p_i^*,p_{-i}^c)-\pi^*=(p^*-c)d(p^c-p^*),
\]
and
\[
\pi^c=\pi(p^c,p^c)=(p^c-c)\big(a-(b-d)p^c\big).
\]
Therefore,
\[
\bar\delta(p^c)=\frac{\pi(p_i^*,p_{-i}^c)-\pi^c}{\pi(p_i^*,p_{-i}^c)-\pi^*} = 1-\frac{\pi^c-\pi^*}{(p^*-c)d(p^c-p^*)} = 1- \frac{1}{(p^*-c)d} \cdot [-(b-d)(p^c+p^*)+\big(a+(b-d)c\big)].
\]
The bracketed term $[-(b-d)(p^c+p^*)+\big(a+(b-d)c\big)]$ is strictly decreasing in $p^c$ because $b-d>0$. Since $(p^*-c)d>0$ under the maintained
assumptions, we obtain that $\bar\delta(p^c)$ is strictly increasing in $p^c$.
\end{enumerate}

\end{proof}

\begin{proof}[\textit{Proof of Proposition \ref{prop:asymmetric_info_threshold}}]
\begin{enumerate}
    \item[(i)] Fix any target collusive price $p^c\in(p^*,p^M]$. Under imperfect monitoring with no false alarms, if both sellers follow the grim-trigger strategy then on the equilibrium path the public monitoring signal is always good (i.e., $s^t=G$ for all $t$) and collusion at $p^c$ continues forever. Hence seller $i$'s on-path present discounted profit stream is
    \begin{equation}\label{eq:asymmetric_collusive_profit}
        \pi^c+\delta_i\pi^c+\delta_i^2\pi^c+\cdots=\frac{\pi^c}{1-\delta_i}.
    \end{equation}
    If seller $i$ deviates in the current period, which means seller $i$ undercuts to $p_i=p^*$ while seller $-i$ maintains $p_{-i}=p^c$, then seller $i$ earns $\pi(p_i^*,p_{-i}^c)$ in the current period. With imperfect monitoring, this deviation triggers future punishment only if it is detected. Since a bad public signal is realized (and observed by both sellers) with probability $\rho_{-i}$, and otherwise the signal remains good with probability $1-\rho_{-i}$. If detected, the bad signal triggers the punishment path and payoffs revert to the competitive equilibrium payoff $\pi^*$ forever starting next period. If not detected, punishment is not triggered and the collusive path continues, yielding $\pi^c$ forever starting next period. Thus seller $i$'s expected present discounted payoff from deviation is
    \begin{equation}\label{eq:asymmetric_deviation_profit}
        \pi(p_i^*,p_{-i}^c)
        +\frac{\delta_i}{1-\delta_i}\Big(\rho_{-i}\pi^*+(1-\rho_{-i})\pi^c\Big).
    \end{equation}

    The incentive compatibility constraint for seller $i$ to not deviate is therefore \eqref{eq:asymmetric_collusive_profit} $\ge$ \eqref{eq:asymmetric_deviation_profit}, which rearranges to
    \[
    \delta_i \ge
    \frac{\pi(p_i^*,p_{-i}^c)-\pi^c}
    {\pi(p_i^*,p_{-i}^c)-\pi^c+\rho_{-i}(\pi^c-\pi^*)}
    \equiv \bar\delta_i(p^c,\rho_{-i}).
    \]
    Hence, infinite price collusion can be sustained as a subgame perfect equilibrium if and only if
    $\delta_i\ge \bar\delta_i(p^c, \rho_{-i})$ for $i=1,2$.

    \item[(ii)] Differentiating $\bar\delta_i(p^c,\rho_{-i})$ with respect to $\rho_{-i}$, we obtain
    \[
    \frac{\partial \bar\delta_i(p^c,\rho_{-i})}{\partial \rho_{-i}}
    =
    -\frac{\big(\pi(p_i^*,p_{-i}^c)-\pi^c\big)\big(\pi^c-\pi^*\big)}
    {\Big(\pi(p_i^*,p_{-i}^c)-\pi^c+\rho_{-i}(\pi^c-\pi^*)\Big)^2}.
    \]
    Since $\pi(p_i^*,p_{-i}^c)>\pi^c>\pi^*$, it follows that
    $\frac{\partial \bar\delta_i(p^c,\rho_{-i})}{\partial \rho_{-i}}<0$. Hence $\bar\delta_i(p^c,\rho_{-i})$ is strictly decreasing in $\rho_{-i}$.
\end{enumerate}
\end{proof}

\newpage
\section{Robustness}\label{sec: robustness}

\subsection{Robustness to Closed-Source and Anti-Collusion Prompting}\label{sec: closed_source_robustness}

\begin{table}[ht!]
\centering
\caption{Summary of Robustness Results (GPT-5-mini)}
\label{tab:results_gpt}
\footnotesize
\begin{tabular}{lcccc}
\toprule
\textbf{Condition} & \textbf{Rounds to} & \textbf{Avg.\ Price} & \textbf{\% Price Elevation} \\
 & \textbf{Convergence (SD)} & \textbf{(SD)} & \textbf{vs.\ Competitive} \\
\midrule
\multicolumn{4}{l}{\textit{Homogeneous Benchmarks (Section~\ref{subsec: two_patient_sellers})}} \\
Two Myopic LLMs  & 419.0 (216.0) & 1.046 (0.045) & $-29\%$ \\
Two Patient LLMs  & 227.8 (134.5) & 2.312 (0.029) & $+57\%$ \\
\midrule
\multicolumn{4}{l}{\textit{Main Heterogeneity Dimensions (predicted by model)}} \\
Patient vs.\ Myopic LLMs (\S\ref{subsec: one_patient_one_myopic}) & 194.2 (133.9) & 1.999 (0.032) & $+36\%$ \\
High Data vs.\ Low-Data LLMs (\S\ref{subsec: information_env_heterogeneity}) & 292.6 (115.0) & 2.087 (0.031) & $+42\%$ \\
\midrule
\multicolumn{4}{l}{\textit{Explicit Anti-Collusion Prompting}} \\
Two Patient LLMs (Anti-Collusion)  & 294.2 (148.3) & 1.044 (0.039) & $-29\%$ \\
Two Myopic LLMs (Anti-Collusion)  & 282.0 (238.4) & 1.033 (0.023) & $-30\%$ \\
Patient vs.\ Myopic LLMs (Anti-Collusion)  & 555.2 (289.8) & 1.052 (0.035) & $-28\%$ \\
High Data vs.\ Low-Data LLMs (Anti-Collusion)  & 622.8 (340.8) & 1.444 (0.046) & $-2\%$ \\
\bottomrule
\end{tabular}
\vspace{0.5em}
\begin{minipage}{0.95\textwidth}
\footnotesize
\textit{Note:} Results from 5 independent runs (48 compute hours) per condition (up to 1,000
periods each).
\end{minipage}
\end{table}


\end{document}